\documentclass[runningheads]{llncs}
\usepackage[T1]{fontenc}
\usepackage{graphicx}
\makeatletter
\def\@seccntformat#1{\@ifundefined{#1@cntformat}%
   {\csname the#1\endcsname\quad}  
   {\csname #1@cntformat\endcsname}
}
\let\oldappendix\appendix 
\renewcommand\appendix{%
    \oldappendix
    \newcommand{\section@cntformat}{\appendixname~\thesection\quad}
}
\makeatother

\usepackage{ amsthm, amssymb, amsmath, stmaryrd, appendix,  adjustbox }
\usepackage{cleveref} 

\usepackage{breakcites}
\usepackage{multirow}

\newtheorem*{theorem*}{Theorem}
\newtheorem*{lemma*}{Lemma}
\newtheorem*{proposition*}{Proposition}

\usepackage{graphicx}
\usepackage{bussproofs}
\usepackage{longtable}
\newcommand{\ket}[1]{|#1\rangle}
\newcommand{\bra}[1]{\langle #1|}
\newcommand{\htriple}[3]{\{ #1 \} #2 \{ #3 \}}

\newcommand{\sem}[1]{[\![ #1]\!]}
\newcommand{\s}[1]{\{#1\}}
\newcommand{\C}{{\mathbb C}}
\newcommand{\gen}[1]{\langle#1\rangle}
\newcommand{\Tr}{\textbf{Tr}}
\newcommand{\WP}{ WP}
\newcommand{\wpd}{wp}
\newcommand{\pt}{pt}

\newlength{\localh}
\newlength{\locald}
\newbox\mybox
\def\mp#1#2{\scalebox{0.8}{\setbox\mybox\hbox{#2}\localh\ht\mybox\locald\dp\mybox\addtolength{\localh}{-\locald}\raisebox{-#1\localh}{\box\mybox}}}

\begin{document}

\title{On the Relative Completeness of Satisfaction-based Quantum Hoare Logic}

%
%
\author{Xin Sun  \and
Xingchi Su  \and
Xiaoning Bian  \and
Huiwen Wu 
}
\authorrunning{X. Sun et al.}
%
\institute{Zhejiang Lab, Hangzhou, China}
\maketitle              

\begin{abstract}

Quantum Hoare logic (QHL) is a formal verification tool specifically designed to ensure the correctness of quantum programs. There has been an ongoing challenge to achieve a relatively complete satisfaction-based QHL with while-loop since its inception in 2006. This paper presents a solution by proposing the first relatively complete satisfaction-based QHL with while-loop. The completeness is proved in two steps. First, we establish a semantics and proof system of Hoare triples with quantum programs and deterministic assertions. Then, by utilizing the weakest precondition of deterministic assertion, we construct the weakest preterm calculus of probabilistic expressions. The relative completeness of QHL is then obtained as a consequence of the weakest preterm calculus. Using our QHL, we formally verify the correctness of Deutsch's algorithm and quantum teleportation.

\keywords{ Hoare logic \and quantum program \and  relative completeness \and   weakest precondition}
\end{abstract}

\section{Introduction}

In the realm of computer science, Hoare logic serves as a tool for formally verifying the correctness of computer programs. It was first established by C. A. R. Hoare in his groundbreaking paper \cite{hoare1969axiomatic} published in 1969, and later expanded upon in \cite{Hoare71}. When executing a program, it's important to consider the conditions that must be met before and after executing a command. The precondition describes the property that the command is based on, while the postcondition describes the desired outcome of the command after it has been correctly executed. This is the key idea of Hoare logic which has become one of the most influential tools in the formal verification of programs in the past decades. It has been successfully applied in the analysis of deterministic \cite{hoare1969axiomatic,Hoare71,Winskel93}, non-deterministic \cite{Dijkstra75,Dijkstra76,Apt84}, recursive \cite{Hoare71,FoleyH71,AptBO09}, probabilistic \cite{Ramshaw79,Hartog02,chadha2007reasoning,rand2015vphl}, and quantum programs \cite{Ying2011Floyd-Hoare,LiuZWYLLYZ19,Unruh19,ZhouYY19,DengF22}. A comprehensive review of Hoare logic is referred to Apt, Boer, and Olderog \cite{AptBO09,apt2019fifty}.

Quantum Hoare logic (QHL) is an extension of the classical Hoare logic. It is specifically designed to ensure the correctness of quantum programs. To make QHL work, a programming language that supports quantum commands, such as unitary operations and measurements, as well as other classical commands is required. The assertion language used in QHL describes quantum states, such as unit vectors and density operators, along with the probabilistic distributions over them. Hoare triples are used to formalize how quantum programs change the states by showing the update on assertions. The proof system characterizes the reasoning about Hoare triples so as to expose the deductive relationship between the state changes caused by quantum programs. There has been a lot of research on QHL that can be classified into two approaches: expectation-based and satisfaction-based. Table \ref{tbl:comparison} compares them in terms of different types of programming languages, including the classical, probabilistic, and Quantum programs. Note that, $\mathcal D (\mathcal H)$ and $\mathcal{P(\mathcal{H})}$ in row Quantum-expectation stand for the set of density operators and Hermitian operators over Hilbert space $\mathcal H$ with certain conditions, following the notations in \cite{Ying2011Floyd-Hoare}.

The first QHL was developed by Chadha \textit{et al} \cite{ChadhaMS06}, which is considered to be the first satisfaction-based construction. Kakutani \cite{Kakutani09} established another satisfaction-based QHL for a functional quantum programming language designed by Selinger \cite{Sel2004-qpl}. Unruh \cite{Unruh2019QHL} provided a satisfaction-based QHL to feature ``ghost variables'', which is used to extend the expressive power of assertions. 

Following the expectation-based approach, several different QHLs have been developed by Ying et al in their series of work \cite{Ying2011Floyd-Hoare,Ying2019TowardAutoVerification,ZhouYY19,Feng2021quantum-classical}. They investigate pure quantum language and infinite-dimensional quantum variables, new auxiliary proof rules, quantum-classical language, and restricted assertions, respectively.

\begin{table}
    \caption{Comparison of the basic notions in different language paradigms.}
    \label{tbl:comparison}
    \centering
\scalebox{0.96}{
    \begin{tabular}{|c|c|c|c|}
    \hline
                    & State                 & Semantic space of   & Satisfiability of\\
                    &                       & assertions          & an assertion $a$\\
    \hline 
    \hline
    \multirow{2}{*}{Classical}      & set of evaluations & subsets of states                 & \multirow{2}{*}{$S \in \sem{a}$} \\
                                    & $S \in \mathbb{S}$    & $\mathbb S \rightarrow \s{0,1}$   & \\
    \hline
    Probabilistic   & (sub)distributions                    & random variables              & \multirow{2}{*}{$\sum_{S \in \mathbb S} \sem{a}(S)\mu(S)$}\\
    (expectation)   & $\mu \in \mathbb S \rightarrow [0,1]$ & $\mathbb S \rightarrow [0,1]$ &\\
    \hline
    Probabilistic   & (sub)distributions                    & subsets of distributions                              & \multirow{2}{*}{$\mu \in \sem{a}$}\\
    (satisfaction)  & $\mu \in \mathbb S \rightarrow [0,1]$ & $ (\mathbb{S} \rightarrow [0,1]) \rightarrow \{0,1\}$ &\\
    \hline
    Quantum         & cq-state                                                      & cq-assertion                                                           & \multirow{2}{*}{$\sum_{S \in \mathbb S} \Tr [\Delta(S) \Theta(S)]$}\\
    (expectation)   & $\Delta \in \mathbb{S} \rightarrow \mathcal{D}(\mathcal{H})$  & $\Theta \in \mathbb{S} \rightarrow \mathcal{P}(\mathcal{H})$ & \\
    \hline
    Quantum         & distributions of cq-state                             & subsets of distributions                                                                & \multirow{2}{*}{$\mu \in \sem{a}$}\\
    (satisfaction)  & $\mu \in (\mathbb{S} \,\times\, \mathcal{H}) \rightarrow[0,1]$  & $((\mathbb{S} \,\times\, \mathcal{H} )\rightarrow[0,1]) \rightarrow \{0,1\}$    & \\
    \hline
    \end{tabular}}

    \normalsize
\end{table}

Several variants of QHL have been verified and  implemented. The QHLProver \cite{LiuZWYLLYZ19} was the first to formalize QHL in a proof assistant and used the verified implementation to verify the correctness of Grover's algorithm. CoqQ \cite{Zhou2023} verifies the soundness of QHL adapted from \cite{Ying2011Floyd-Hoare, Ying2019TowardAutoVerification} in the Coq proof assistant. The work then uses CoqQ to verify several quantum algorithms, including Grover's algorithm and quantum phase estimation.


The problem of establishing the relative completeness of  satisfaction-based QHL with While-loop has remained unresolved since its inception in 2006 \cite{ChadhaMS06}. This paper aims to address this issue by developing the first satisfaction-based QHL with while-loop that is relatively complete. To accomplish this, we follow the same approach as used to demonstrate the relative completeness of satisfaction-based probabilistic Hoare logic (PHL) \cite{SunSBLC24}. Our solution involves two steps. Firstly, we establish a semantics and proof system of Hoare triples with quantum programs and deterministic assertions. Then, using the weakest precondition of deterministic assertion, we construct the weakest preterm calculus of probabilistic expressions. The relative completeness of QHL is then obtained as a consequence of the weakest preterm calculus.

This paper is organized as follows: Section \ref{sec:deterministic} constructs a relatively complete QHL with deterministic assertion. 
 Section \ref{sec:probabilistic} extends the result of the QHL with deterministic assertion to the QHL with probabilistic assertion.   Section \ref{sec:conclusion} concludes the paper and discusses possible future work. The Appendix introduces some basics of quantum computing, gives proof details of major theorems, and provides two applications of our QHL: correctness proofs of Deutsch's algorithm and quantum teleportation.

\section{Quantum Hoare logic with deterministic assertion}
\label{sec:deterministic}
 
\subsection{Deterministic expressions and formulas}
 
Let $\mathbb{PV}= \{X,Y,Z,\ldots\}$ be a set of classical program variables. Let $\mathbb{LV}= \{x,y,z,\ldots\}$ be a set of logical variables.  Program variables are those variables that may occur only in programs. They constitute deterministic expressions. Deterministic expressions are classified into arithmetic expression $E$ and Boolean expression $B$. The arithmetic expression consists of integer constants $n \in \mathbb{Z}$, variables from $\mathbb{PV}$ and arithmetic operators $\{+,-,\times,...\}\subseteq \mathbb{Z}\times\mathbb{Z}\rightarrow \mathbb{Z}$. 

\begin{definition}[Arithmetic expressions]
     Given a set of program variables $\mathbb{PV}$, we define the arithmetic expression $E$ as follows:

\begin{center}
    $E:=n\mid X\mid (E\ aop\ E)$.
\end{center}
\end{definition}

\noindent
An arithmetic expression ($E$) can take the form of an integer constant ($n$), a program variable ($X$), or a combination of two arithmetic expressions ($E \ aop \ E$) composed by an arithmetic operation ($aop$). 

The Boolean constant set is $\mathbb{B}=\{\top,\bot\}$. We use relational operators ($rop$) {\em i.e.,} $\{>,<,\geq,=, ...\}\subseteq\mathbb{Z}\times\mathbb{Z}\rightarrow \mathbb{B}$ to compare arithmetic expressions ($E$) which represent integers in programs. 


\begin{definition}[Boolean expressions] The Boolean expression is defined as follows:

 \begin{center}
     $B:=\top\mid \bot\mid (E\ rop\ E)\mid \neg B\mid (B\wedge B).$
 \end{center}
 
\end{definition}
 
A Boolean expression is a statement that can only be true or false. In the expression $(E\ rop\ E)$, the truth value is determined by the binary relation $rop$ between two integers. 

The semantics of deterministic expressions are defined on deterministic states, which are mappings $S: \mathbb{PV}\rightarrow\mathbb{Z}$. Let $ \mathbb{S}$ be the set of all deterministic states.
Each state $S \in  \mathbb{S}$ is a description of the value of each program variable. Consequently, the semantics of the arithmetic expressions is $[\![E]\!]:\ \mathbb{S}\rightarrow\mathbb{Z}$ which assigns an integer value to each deterministic state. Similarly, the semantics of Boolean expressions is $[\![B]\!]:\ \mathbb{S}\rightarrow\mathbb{B}$ that maps each state to a Boolean value. 

\begin{definition}[Semantics of deterministic expressions]\label{semanticsdeterministicexpression}
The semantics of deterministic expressions are defined inductively as follows:

\begin{center}
    \begin{tabular}{rll}
      $[\![X]\!]S$ & = & $S(X)$ \\
      $[\![n]\!]S$ & = & $n$ \\
      $[\![E_{1}\ aop\ E_{2}]\!]S$ & = & $[\![E_{1}]\!]S\ aop\ [\![E_{2}]\!]S$\\
      $[\![\top]\!]S$ & = & $\top$\\
      $[\![\bot]\!]S$ & = & $\bot$\\
      $[\![E_{1}\ rop\ E_{2}]\!]S$ & = & $[\![E_{1}]\!]S\ rop\ [\![E_{2}]\!]S$\\
      $[\![\neg B]\!]S$ & = & $\neg[\![B]\!]S$\\
      $[\![B_{1}\wedge  B_{2}]\!]S$ & = & $[\![B_{1}]\!]S\wedge\  [\![B_{2}]\!]S$
    \end{tabular}
\end{center}
\end{definition}


\subsection{Commands}

We use a  command called \textit{random assignment} that allows programs to exhibit probabilistic behaviors and use quantum operations such as unitary operation and measurement to exhibit quantum computation. Let $\mathbb{QV} =\{q_1,\ldots, q_m\}$ be a finite set of  quantum variables.

\begin{definition}[Syntax of command expressions]\label{syntaxcommandexpression}
    The commands are defined inductively as follows: 
    \vspace{-2mm}
\begin{center}
    $C:=skip\mid X\leftarrow E\mid X\xleftarrow[]{\$}R\mid C_{1};C_{2}\mid if\ B\ then\ C_{1}\ else\ C_{2}  \mid while \mbox{ } B \mbox{ } do \mbox{ } C   \mid $
 
\end{center}

\begin{center}

$   U[\overline{q}] \mid X  \leftleftarrows q_i$
\end{center}

\noindent
where $R=\{a_1: r_1,\cdots,a_n: r_n\}$ in which  $\{ r_1,\cdots, r_n\}$ is a set of integers and $a_1, \ldots , a_n$ are real numbers such that $0<a_i <1$ and $a_1+ \ldots+a_n =1$.
\vspace{-2mm}
\end{definition}

The command $skip$ represents a null command that does not do anything. $X\leftarrow E$ is the normal assignment. $X\xleftarrow[]{\$}R$ is the  random assignment which can be read as a value $r_i$ is selected randomly with probability $a_i$ and assigned to $X$. $C_1;C_2$ stands for the sequential execution of  $C_1$ and $C_2$. The next two expressions are the classical conditional choice and loop, respectively. In terms of quantum commands, $ U[\overline{q}]$ applies the unitary operator $U$ to the quantum variables indicated by $   \overline{q}$. The command $ X  \leftleftarrows q_i $ measures the quantum variables $q_i$ using the computational basis and stores the measurement result in the classical variable $X$. The basics of quantum computation can be referred to in Appendix A.

For every quantum variable $q_i$, we assume a two-dimensional Hilbert space $\mathcal{ H}_{i}$. Let $\mathcal{ H}  = \mathcal{ H}_{1} \otimes \ldots \otimes \mathcal{ H}_{m} $ be the Hilbert space of all quantum variables. A quantum state is a unit vector $|u\rangle \in \mathcal{ H}$. We assume that two unit vectors,  which differ only in a globe phase, represent the same quantum state. A pure classical-quantum state (cq-state) is a pair $\theta=(S, |u\rangle ) $ of a classical state and a quantum state. We use $\theta_c$ to denote its classical part and $\theta_q$ to denote its quantum part. That is, we let $\theta_c = S$ and $\theta_q= |u\rangle$.

Let $\mathbb{PS}$ be the set of all pure cq-states. A probabilistic cq-state (also called mixed cq-state), denoted by $\Theta$, is a probability sub-distribution on pure cq-states, i.e. $\Theta: \mathbb{PS}\to [0,1]$ such that $\Sigma_{\theta \in\mathbb{PS}} \Theta(\theta)\leq 1$. It is worth noting that we utilize subdistributions considering the cases where some programs may never terminate. We use $D(\mathbb{PS})$ to denote the set of all probabilistic cq-states.

For an arbitrary pure cq-state $ \theta \in\mathbb{PS}$, $\widehat{\theta} $ denotes a special probabilistic cq-state which maps $\theta$ to 1 and maps any other state to 0, {\em i.e.}, a point distribution. 


\begin{definition}[Semantics of command expressions]\label{def.semanticsofcommand}
The semantics of command expressions is a function $[\![C]\!]\in D(\mathbb{PS})\rightarrow D(\mathbb{PS})$. It is defined inductively as follows:

\begin{itemize}

\item  $[\![skip]\!](\Theta)=\Theta$

\item $[\![X\leftarrow E]\!](\widehat{\theta})=    \widehat{\theta'}  $, where $\theta'_c = \theta_c[X \mapsto [\![E]\!]\theta_c ]$, $\theta'_q = \theta_q $.

\item $[\![X\leftarrow E]\!]( \Theta )=    \displaystyle\sum\limits_{\theta \in \mathbb{PS}}  \Theta(\theta)\cdot   \widehat{\theta'}$, where $\widehat{\theta'}= [\![X\leftarrow E]\!](\widehat{\theta})$.

\item $[\![X\xleftarrow{\$}\{a_1:r_{1},...,a_n:r_{n}\}]\!](\Theta)=  \displaystyle\sum_{i=1}^{n} a_i [\![X\leftarrow r_{i}]\!](\Theta)$.

\item $[\![C_{1};\ C_{2}]\!](\Theta)=[\![C_{2}]\!]([\![C_{1}]\!](\Theta))$

\item   $[\![if\ B\ then\ C_{1}\ else\ C_{2}]\!](\Theta)=[\![C_{1}]\!](\downarrow_{B}(\Theta))+[\![C_{2}]\!](\downarrow_{\neg B}(\Theta))$.

\item $[\![   while \mbox{ } B \mbox{ } do \mbox{ } C  ]\!] (\Theta ) =  \displaystyle\sum_{i = 0}^{\infty} \downarrow_{\neg B}  (  (  [\![   C   ]\!] \circ \downarrow_{ B})^i (\Theta ) ) $.


\item $[\![  U[\overline{q}]  ]\!](\widehat{\theta})=\widehat{\theta'}$, where  $\theta'_c=\theta_c$,$ \theta_q' =U_{\overline{q}}\theta_q$.

\item $[\![  U[\overline{q}]  ]\!](\Theta)=\displaystyle\sum\limits_{\theta\in\mathbb{PS}} \Theta(\theta)\cdot [\![  U[\overline{q}]  ]\!](\widehat{\theta})$.

\item $[\![   X  \leftleftarrows q_i    ]\!] (\widehat{\theta}) =  \displaystyle\sum\limits_{j=0}^{1}    p_j \widehat{\theta_j}   $, where $p_j=\Tr( ( |j\rangle \langle j| \otimes I_{-i})   \theta_q    \theta_q^{\dagger}         ) $, $\theta_{j, c}= \theta_c[X \mapsto j ]$, $\theta_{j, q}= \frac{( |j\rangle \langle j| \otimes I_{-i}) \theta_q      }{  \sqrt{p_j}   }  $ if $p_j \neq 0$, otherwise $\theta_{j, q}=  (|j\rangle \langle j| \otimes I_{-i}) \theta_q     $.

\item $[\![   X  \leftleftarrows q_i    ]\!] (\Theta)=\displaystyle\sum\limits_{\theta\in\mathbb{PS}}\Theta(\theta)\cdot [\![  U[\overline{q}]  ]\!](\widehat{\theta})$.

\end{itemize}

\end{definition}


The command $skip$ changes nothing. We write $\theta_c  [X\mapsto[\![E]\!]\theta_c]$ to denote the classical state that assigns the same values to the variables as $\theta_c$, except the variable $X$ is assigned the value $[\![E]\!]\theta_c$.  Here, $\downarrow_{B}(\Theta )$ denotes the distribution $\Theta $ restricted to those pure states that make $B$ true. Formally, $\downarrow_{B}(\Theta )= \Lambda  $ with $\Lambda(\theta)= \Theta(\theta)$ if $[\![B]\!]\theta_c =\top$ and $\Lambda( \theta )=0$ otherwise. So the condition choice can split any $\Theta $ into two parts {\em with respect to} a boolean expression $B$, one of which is the sub-distribution of states that satisfy $B$ and then execute the command $C_{1}$, the remaining execute the command $C_{2}$.


\subsection{Proof system with deterministic assertions}\label{sec.proofsystemdeter}

Now we define deterministic assertions based on deterministic expressions. 

\begin{definition}[Syntax of deterministic assertions]\label{def.syntaxdeterministicformulas}
The deterministic assertions are defined by the following BNF:
\begin{center}
$\phi:= \top\mid \bot  \mid P_j^0 \mid P_j^1  \mid  (e\ rop\ e) \mid \neg\phi\mid (\phi\wedge  \phi)\mid \forall x\phi  \mid  [U_{\overline{q}}] \phi \mid   [ \mathfrak{P}_j^0] \phi \mid  [ \mathfrak{P}_j^1] \phi  \mid \bigwedge\limits_{i=0}^{\infty} \phi_i$
\end{center}
where $e$ represents arithmetic expression build on $\mathbb{LV} \cup \mathbb{PV}$:

\begin{center}
    $e:=n\mid X\mid x\mid (e\ aop\ e)$.
\end{center}

\end{definition}

Here, $\top$ and $\bot$ are boolean constants which represent true and false, respectively.  $P^i$ is the projector $|i\rangle \langle i |$.  
We restrict the logical operators to the classical operators $\neg$ and $\wedge$ since $\vee$ and $\to$ can be expressed in the standard way. The formula $\forall x\phi$ applies a universal quantifier to the logical variable $x$ in formula $\phi$. $[U_{\overline{q}}] \phi$ intuitively means that $\phi$ is true after the unitary operation $U$ is applied on the qubits $\overline{q} $. $ [ \mathfrak{P}_j^i] \phi$ means that $\phi$ is true after the qubit $q_j$ is projected to state $|i\rangle$. $\bigwedge_{i=0}^{\infty} \phi_i$ is the infinitary  conjunction. We will need an infinitary conjunction to define the weakest preconditions.

An interpretation $I: \mathbb{LV} \mapsto \mathbb{Z}$ is a function that maps logical variables to integers. Given an interpretation $I$ and a pure cq-state $\theta$, the semantics of $e$ is defined as follows.

\begin{center}
    \begin{tabular}{rll}
     $[\![n]\!]^I \theta$ & = & $n$ \\
      $[\![X]\!]^I  \theta$ & = & $\theta_c(X)$ \\
      $[\![x]\!]^I  \theta$ & = & $I(x)$ \\
      $[\![E_{1}\ aop\ E_{2}]\!]^I \theta$ & = & $[\![E_{1}]\!]^I \theta\ aop\ [\![E_{2}]\!]^I \theta$\\
    \end{tabular}
\end{center}

\noindent
The semantics of a deterministic assertion, to be defined in the next definition, is denoted by $[\![\phi]\!]^I=\{\theta :  \theta \models^I\phi\}$ which represents the set of all states satisfying $\phi$. 

\begin{definition}[Semantics of deterministic assertions]\label{def.semanticsdeterministicformula}
The semantics of deterministic assertions is defined inductively as follows:

\begin{center}
\begin{tabular}
{p{.15\textwidth}p{.04\textwidth}p{.7\textwidth}}
$[\![\top]\!]^I$ & = & $\mathbb{PS}$\\
$[\![\bot]\!]^I$ & = & $\emptyset$\\ 

$[\![P_j^i]\!]^I$ & =& $\{  \theta \in\mathbb{PS}\mid  \Tr(P^i \Tr_{-j} ( \theta_q \theta_q^{\dagger}) ) =1  \} $\\ 
$[\![e_{1}\ rop\ e_{2}]\!]^I$ & = & $\{\theta  \in\mathbb{PS}\mid [\![e_{1}]\!]^I \theta_c\ rop\ [\![e_{2}]\!]^I \theta_c=\top\}$\\
$[\![\neg\phi]\!]^I$ & = & $\mathbb{PS}\backslash[\![\phi]\!]^I$\\
$[\![\phi_{1}\wedge\phi_{2}]\!]^I$ & = & $[\![\phi_{1}]\!]^I\cap[\![\phi_{2}]\!]^I$\\
$[\![\forall x\phi]\!]^I$ &  $=$ & $\{ \theta  \mid  $ for all integer $n$ and $I'=I[x \mapsto n]$, $\theta\in [\![ \Phi]\!]^{I'} \}$\\
$[\![ [U_{\overline{q}}] \phi    ]\!]^I$ &  $=$ & $\{ \theta  \mid  U_{\overline{q}} \theta \models^I \phi \} $, here $  U_{\overline{q}} \theta$ is short  for $(\theta_c,  U_{\overline{q}} \theta_q)$ \\
$[\![ [ \mathfrak{P}_j^i] \phi   ]\!]^I$ &  $=$ & $\{ \theta  \mid    \theta' \models^I \phi, $  where $\theta'_c=\theta_c$,  $\theta'_q=\frac{(P^i_j \otimes I_{-j}) \theta_q      }{  \sqrt{p_i}   } $, $p_i= \Tr( (P^i_j \otimes I_{-j})   \theta_q   \theta_q^{\dagger}         )  $ and $p_i > 0 $   $ \} $ \\


$[\![ \bigwedge\limits_{i=0}^{\infty} \phi_i   ]\!]$ &  $=$ &  $ \bigcap_{i=0}^{\infty} [\![\phi_{i}]\!]^I$
        \end{tabular}
    \end{center}
\end{definition}


\noindent
Most of the above cases are self-explanatory. $[U_{\overline{q}}]\phi$ is true on these pure cq-states which renders $\phi$ true after applied the unitary operator $U$ on quantum variables $\overline{q}$. $[ \mathfrak{P}_j^i] \phi$ is satisfied on a pure cq-state $\theta$ if and only if after  projecting variable $q_i$ in $\theta_q$ to the quantum state $|i\rangle$, the obtained pure cq-state $\theta'$ satisfies $\phi$.

Definition~\ref{def.semanticsdeterministicformula} gives the semantics of deterministic formulas over pure cq-states. We straightforwardly extend the semantics to probabilistic states as follows:


\begin{center}
    $\Theta  \models^I \phi$   iff   for each support $\theta$ of $\Theta$, $\theta \in [\![\phi]\!]^I$.
\end{center}

\noindent 
We call it possibility semantics because the definition intuitively means that a deterministic formula $\phi$ is true on a probabilistic cq-state $\Theta$ if and only if $\phi$ is true on all possible pure cq-state. 


A QHL proof system consists of Hoare triples. A Hoare triple $\{\phi\}C\{\psi\}$ is valid if for each pure cq-state that satisfies $\phi$, after the termination of the execution command $C$, the mixed consequence of the cq-state satisfies $\psi$. Formally,  

\begin{center}
$\models \{\phi\}C\{\psi\}$ if for each interpretation $I$ and each pure cq-state $\theta$, if $\theta \models^I  \phi$, then $[\![C]\!](\widehat{\theta}) \models^I \psi$.
\end{center}


\noindent
We now build a proof system (QHL$_d$) for the derivation of Hoare triples with deterministic assertions. 

\begin{definition}[Proof system QHL$_d$]
    The proof system  QHL$_d$ consists of the following inference rules:

    \begin{center}
        \begin{tabular}{ll}
         $SKIP:$ & $\frac{}{\vdash\{\phi\}\texttt{skip}\{\phi\}}$\\
         \\
         $AS:$ & $ \frac{}{\vdash \{   \phi[X/ E]   \}  X \leftarrow E\{ \phi  \} }$\\
         \\
         $PAS:$ & $   \frac{}{\vdash \{   \phi[X/k_1] \wedge \ldots \wedge  \phi[X/k_n]   \}  X\xleftarrow{\$}\{a_1: k_{1},...,a_n: k_{n} \}   \{ \phi  \} }$\\  
         \\
         $SEQ:$ & $ \frac{\vdash\{\phi\}C_{1}\{\phi_{1}\}\quad \vdash\{\phi_{1}\}C_{2}\{\phi_{2}\}} {\vdash\{\phi\}C_{1};C_{2}\{\phi_{2}\}} $\\
        \\
          $IF:$ & $ \frac{\vdash\{\phi\wedge B\}C_{1}\{\psi\}\quad \vdash\{\phi\wedge\neg B\}C_{2}\{\psi\}} {\vdash\{\phi\} \texttt{if}\ B\ \texttt{then}\ C_{1}\ \texttt{else}\ C_{2}\{\psi \} } $\\          \\
           $WHILE:$ & $\frac{\vdash \{\phi  \wedge B\}   C   \{\phi\}    } {\vdash\{\phi  \}    \texttt{while} \mbox{ }B \mbox{ } \texttt{do} \mbox{ }C   \{\phi \wedge \neg B    \}}$\\          
           \\
           $UNITARY:$ & $\frac{     } {\vdash\{ [U_{\overline{q}}] \phi  \}     U[\overline{q}]     \{\phi    \}}$\\
                \\
           $MEASURE:$ & $\frac{    } {\vdash\{  (  [ \mathfrak{P}^0_j ] \phi[X/0] \vee P^1_j ) \wedge   ( [ \mathfrak{P}^1_j ] \phi[X/1] \vee P^0_j )      \}    X  \leftleftarrows  q_j     \{\phi    \}}$\\

           	  \\
          $CONS:$ & $ \frac{\models\phi'\rightarrow\phi\quad  \vdash\{\phi\}C\{\psi\}\quad \models\psi\rightarrow\psi'}  {\vdash\{\phi'\}C\{\psi'\}}$\\
          \\

        \end{tabular}
    \end{center}
\end{definition}

\subsection{Soundness and completeness}

\begin{theorem}[Soundness]\label{th.soundnessphl}
 The proof system QHL$_d$ is sound, $i.e.$, for all deterministic formula $\phi$ and $\psi$ and command $C$, $\vdash\{\phi\}C\{\psi\}$ implies $\models\{\phi\}C\{\psi\}$. 
\end{theorem}
\begin{proof}
Please refer to Appendix \ref{proof:soundness-d}.
\end{proof}



In order to determine the completeness, we use the method of weakest preconditions denoted by $\wpd(C,\psi)$. Essentially, we want to find the largest set of states (represented as an assertion) from which if a program $C$ is executed, the resulting states satisfy the given postcondition $\psi$, and $[\![ \wpd(C,\psi) ] \!]$ represents this set. This set is considered to be a precondition since $\models \{\wpd(C,\psi) \} C \{ \psi\}$ and it is also the weakest precondition because if $\models \{\phi\}C\{\psi\}$, then $\models \phi \rightarrow \wpd(C,\psi)$. Although the weakest precondition is a semantic notion, we can represent it syntactically using $ \wpd(C,\psi)$. In this paper, we will treat the weakest precondition as a syntactic notion when it is more convenient for our constructions and proofs.


\begin{definition}[Weakest precondions]\label{def.weakestprecondition}
The weakest precondition is defined inductively on the structure of commands as follows:

\begin{enumerate}
\item $\wpd(skip, \phi) = \phi$.
\item  $\wpd(  X \leftarrow E, \phi )= \phi[X/ E ]$.
\item  $\wpd( X\xleftarrow{\$}\{a_1: k_{1},...,a_n: k_{n} \}, \phi  ) =   \phi[X/k_1] \wedge \ldots \wedge  \phi[X/k_n]$.
\item $\wpd( C_1 ; C_2 , \phi) = \wpd(C_1, \wpd(C_2, \phi))$.
\item $\wpd( if\ B\ then\ C_{1}\ else\ C_{2}, \phi   )=  (B \wedge \wpd(C_1, \phi))  \vee (\neg B \wedge \wpd(C_2, \phi))$.
\item $\wpd( while \mbox{ }B \mbox{ } do \mbox{ }C  , \phi)= \bigwedge\limits_{k = 0}^{\infty}  \psi_k   $, 
where $\psi_0 = \top$ and $\psi_{i+1} = ( B \wedge \wpd( C, \psi_{i}))  \vee ( \neg B  \wedge \phi)$.

\item $\wpd(U[\overline{q}], \phi ) = [U_{\overline{q}}] \phi$.

\item $\wpd(X \leftleftarrows  q_j  ,\phi) = (     [\mathfrak{P}^0_j] \phi[X/0] \vee P^1_j    ) \wedge (   [\mathfrak{P}^1_j] \phi[X/1]  \vee  P^0_j  ) $.

\end{enumerate}
\end{definition}

Some readers may be wondering if it is possible to define $\wpd( while \mbox{ }B \mbox{ } do \mbox{ }C, \phi)$ without using an infinitary conjunction. In classical and probabilistic Hoare logic, the weakest precondition of the while-loop can be defined using G$\ddot{o}$del's $\beta$ predicate without an infinitary conjunction. Unfortunately, this does not appear to be possible in the current setting. The reason for this is that in our QHL, the number of pure cq-states is uncountable, whereas in classical and probabilistic Hoare logic, the number of (pure) states is countable. This is why they can be encoded by natural numbers, which is a crucial step in defining the weakest precondition. 

If we want to remove an infinitary conjunction, we have to reduce the number of pure cq-states to countable infinity. This is possible by, for example, restricting unitary operators to the Clifford+T family and letting the set of cq-states be the minimal set that contains $|0\rangle$ and is closed under the Clifford+T operations. However, the rigorous study of this issue will be left for future work.

\begin{theorem}\label{th.precondition}

$\Theta \models^I \wpd(C,\phi)$ iff $[\![ C]\!] (\Theta) \models^I \phi$.

\end{theorem}
\begin{proof}
Please refer to Appendix \ref{proof:wpd}.
\end{proof}


\begin{proposition}\label{prop.weakestprecondition}
It holds that $ \vdash \{   \wpd(C, \phi) \}   C\{\phi\}$.
\end{proposition}
\begin{proof}
Please refer to Appendix \ref{proof:prop-wp}.
\end{proof}



\begin{lemma}

If $\models \{\phi \} C \{\psi \}$ then $\models \phi \rightarrow \wpd(C, \psi) $

\end{lemma}

\begin{proof}

If $\theta \models \phi$, then $[\![ C]\!]\theta\models \psi $. Therefore, $\theta \models \wpd(C,\psi)$.

\end{proof}

\begin{corollary}[Completeness]\label{th.completeness} If $\models \{\phi \} C \{ \psi\}$, then $\vdash \{\phi \} C \{ \psi\}$.

\end{corollary}
 
\begin{proof} 
If $\models \{\phi \} C \{ \psi\}$, then $\models \phi \rightarrow  \wpd(C, \psi)$. By Proposition \ref{prop.weakestprecondition}, we know $\vdash \{ \wpd(C, \psi) \} C \{\psi \}$. Then $\vdash \{\phi \} C \{ \psi\}$ is obtained by applying the inference rule (CONS).  
\end{proof}

We have already established that QHL$_d$ is both sound and complete, but it only deals with deterministic formulas. Probabilistic assertions, on the other hand, are more expressive, as they can describe the probabilistic nature of probabilistic cq-states, such as the probability of $X>3$ being $\frac{1}{2}$. The next section delves into probabilistic assertions and their corresponding proof system. To begin, a useful lemma is presented, highlighting the relationship between the WHILE command and the IF command, which plays a significant role in the completeness proof in the following section.

\begin{lemma}\label{key lemma}
For all $i\geq 0$ and all pure cq-state $\theta$, if $\theta \models^I \neg \wpd(C^0 , \neg B ) \wedge \ldots \wedge  \neg \wpd(C^{i-1} , \neg B) \wedge   \wpd(C^{i} , \neg B ) $, then 
\[
[\![  \texttt{while} \mbox{ } B \mbox{ } \texttt{do} \mbox{ } C  ]\!] (\widehat{\theta}) =  [\![  ( \texttt{if}\mbox{ } B \mbox{ } \texttt{then}\mbox{ }  C \mbox{ }  \texttt{else}\mbox{ }  \texttt{skip} )^{i} ]\!] (\widehat{\theta}).
\]
\end{lemma}
\begin{proof}
    Please refer to Appendix \ref{proof:lemma-wpd}.
\end{proof}

 \section{Quantum Hoare logic with probabilistic assertion}
\label{sec:probabilistic}

\subsection{Probabilistic formulas}

To better describe the probabilistic aspects of cq-states, it is inevitable to expand deterministic formulas into probabilistic formulas. To achieve this, we begin by defining real expressions (also known as probabilistic expressions) which serve as the fundamental components of probabilistic formulas.

\begin{definition}[Real expressions] Let $\mathbb{RV}$ be a set of real variables. The real expression $r$ is defined as follows:

$$r:=  a \mid      \mathfrak{x}      \mid  \mathbb{P}(\phi ) \mid r\mbox{ } aop \mbox{ } r  \mid   \phi \Rightarrow Q   \mid  \sum\limits_{i =0}^{\infty}  r_i   $$
 
\end{definition}

\noindent
Here, $a \in \mathbb{R}$ is a real number,  $\mathfrak{x}   \in \mathbb{RV}$ is a real variable, and $aop$ is arithmetic operations on real numbers such as $+,-, \times  , \ldots $.  $\mathbb{P}(\phi )$ is the probability that a \textit{deterministic assertion} $\phi$ is true. $Q$ is a projective operator on $\mathcal{H} $, which means that $Q$ is self-adjoint and idempotent. 	A \textit{cq-conditional} $ \phi \Rightarrow Q$  denotes the probability of observing $Q$ in those states where $\phi$ is true. $\sum\limits_{i =0}^{\infty}  r_i$ is  an infinite summation. 

To decide the value of logical variable $x$, an interpretation $I$ is used to map logical variables to integers and real variables to real numbers.

\begin{definition}[Semantics of real expressions]  Given an interpretation $I$ and a probabilistic cq-state $\Theta$, the semantics of real expressions are defined inductively as follows.

\begin{itemize}
\item $[\![a ]\!]^I_{\Theta} = a $.

 \item $[\![  \mathfrak{x}     ]\!]^I_{\Theta} = I(  \mathfrak{x}   ) $.

\item $[\![ \mathbb{P}(\phi )  ]\!]^I_{\Theta}  = \sum\limits_{ \theta \models^I \phi} \Theta (\theta) $.

\item $[\![ r_1\mbox{ } aop \mbox{ } r_2   ]\!]^I_{\Theta} =   [\![r_1  ]\!]^I_{\Theta}    \mbox{ } aop \mbox{ }    [\![r_2  ]\!]^I_{\Theta}  $.

\item $ [\![  \phi \Rightarrow Q  ]\!] ^I_{\Theta} = \sum\limits_{\theta \models \phi}  \Theta(\theta)   \Tr( Q \theta_q \theta_q^{\dagger} ) $. 

\item $  [\![    \sum\limits_{i =0}^{\infty}  r_i  ]\!] ^I_{\Theta} =     \sum\limits_{i =0}^{\infty}  [\![      r_i ]\!] ^I_{\Theta} $

\end{itemize}
\end{definition}

A real number $a$ is interpreted as itself on a probabilistic cq-state. The probability that the deterministic formula $\phi$ holds over some probabilistic cq-state is denoted by $\mathbb{P}(\phi)$. This probability can be computed by adding up the probabilities of all pure cq-states where $\phi$ is true. $r_1\ aop\ r_2$ characterizes the arithmetic calculation between two real expressions. Additionally, the summation of infinitely many values of $r_i$s is represented by $\sum\limits_{i =0}^{\infty} r_i$.

Now we define probabilistic formulas (assertions). These formulas typically involve two real expressions and are used to compare probabilities. For instance, it is common to see a probabilistic formula like $\mathbb{P}(X>1)<\frac{1}{2}$, which compares two probabilities. 

\begin{definition}[Syntax of probabilistic formulas]  Probabilistic formulas are defined inductively as follows.

$$\Phi=    (r \mbox{ }  rop \mbox{ }   r) \mid \neg \Phi \mid (\Phi \wedge \Phi)$$

where $r$ is a real expression. 
\end{definition}

For example,  $(\mathbb{P}([U_{\overline{q}}](X>0))>\frac{1}{2})\wedge \neg (\mathbb{P}([ \mathfrak{P}_1^0] (X>1))<(P_1^0\Rightarrow Q))$ is a well-formed probabilistic formula. 

\begin{definition}[Semantics of probabilistic formulas]\label{def.semanticsprobabilisticformulas}
Given an interpretation $I$, the semantics of the probabilistic assertion is defined in the probabilistic cq-states $\Theta$ as follows:  
 
\begin{itemize}
\item $\Theta \models^I   r_1  \mbox{ }  rop \mbox{ }     r_2 $ if  $  [\![r_1]\!]^I_{\Theta} \mbox{ }  rop \mbox{ }   [\![r_2]\!]^I_{\Theta} = \top$ 

\item $\Theta \models^I \neg \Phi$ if not $\Theta \models^I  \Phi$

\item $\Theta \models^I  \Phi_1 \wedge \Phi_2$ if   $\Theta \models^I  \Phi_1$ and $\Theta \models^I  \Phi_2$

\end{itemize} 
\end{definition}

\subsection{Proof system with probabilistic assertions}

Now we delve deeper into the Hoare triples $\{\Phi\}C\{\Psi\}$, specifically those with probabilistic formulas, and the challenge of determining the weakest preconditions for given commands and probabilistic formulas. To address this, we extend the concept of \textit{weakest preterms} in probabilistic Hoare  logic   \cite{chadha2007reasoning,SunSBLC24} to the quantum setting. The concept of weakest preterm is defined as the term whose interpretation on the initial probabilistic cq-state is the same as the interpretation of the given term (real expression) on the resulting cq-state after executing a given command $C$. To define the weakest preterm, we also utilize the notion of conditional term from Chadha \textit{et al.} \cite{chadha2007reasoning,SunSBLC24}.

\begin{definition}[Conditional terms] The conditional term $r/B$ of a real expression $r$ with a Boolean formula $B$ (a deterministic formula) is inductively defined as follows.
 \begin{itemize}
 
\item $a/B :=a $

 \item $ \mathfrak{x}/B := \mathfrak{x} $ 
 
 \item $  \mathbb{P}(\phi ) /B :=  \mathbb{P}(\phi \wedge B)$
 
 \item $(r_1\mbox{ } aop \mbox{ } r_2 ) /B := r_1/B \mbox{ } aop \mbox{ } r_2/B$

 \item $ (\phi \Rightarrow Q)   /B :=   (\phi \wedge B) \Rightarrow Q  $
 
 \item $      (  \sum\limits_{i =0}^{\infty} r_i ) /B=      \sum\limits_{i =0}^{\infty} ( r_i/B)    $
 
 \end{itemize}
\end{definition}

A conditional term intuitively denotes a probability under some condition described by a deterministic formula, which is shown by the next lemma.

 \begin{lemma}
 
Given an arbitrary probabilistic cq-state $\Theta$ and an interpretation $I$, it holds that $ [\![ r /B ]\!]^I_{\Theta} = [\![ r  ]\!]^I_{\downarrow_{B}  \Theta}  $.

 \end{lemma}
 
 \begin{proof}

We only consider the non-trivial cases where $r$ is in the form of $ \mathbb{P}(\phi ) $ or $\phi \Rightarrow Q $:
\begin{itemize}
    \item When $r=\mathbb{P}(\phi)$, $ [\![  \mathbb{P}(\phi )  /B ]\!]^I_{\Theta} = [\![ \mathbb{P}(\phi \wedge B) ]\!]^I_{\Theta}=   [\![  \mathbb{P}(\phi )   ]\!]^I_{\downarrow_{B}  \Theta}  $.
    \item When $r= \phi \Rightarrow Q $, $[\![   (\phi  \Rightarrow Q )/B ]\!]^I_{\Theta}  =  [\![  ( \phi \wedge  B )\Rightarrow Q   ]\!]^I_{\Theta    }  =  \sum\limits_{\theta \models \phi \wedge B}  \Theta(\theta)   \Tr( Q \theta_q \theta_q^{\dagger} ) =  \sum\limits_{\theta \models \phi  }  \downarrow_{B} \Theta(\theta)   \Tr( Q \theta_q \theta_q^{\dagger} ) =   [\![    \phi  \Rightarrow Q   ]\!]^I_{  \downarrow_{B}  \Theta}$.
\end{itemize}
\end{proof}

\subsection{Weakest preterms calculus and weakest precondition} 

We are ready to define the weakest preterms now. Defining the weakest preterms for a real expression $a$, $\mathfrak{x}$, and $r_1\ aop\ r_2$ is straightforward. With regard to $\mathbb{P}(\phi)$ and $ \phi \Rightarrow Q$, we need to divide the cases by different commands.  Those cases with non-quantum commands, especially with  the while-loop, are essentially the same as  their analogues in probabilist Hoare logic in \cite{SunSBLC24}.

\begin{definition}[Weakest preterms] 
The weakest preterm of a real expression $r$ with command $C$ is inductively defined as follows.

\begin{enumerate}

\item $\pt(C,a)=a$

\item $\pt(C,  \mathfrak{x}   )=  \mathfrak{x}   $

\item $\pt(C,r_1\mbox{ } aop \mbox{ } r_2 )= \pt(C, r_1) \mbox{ } aop \mbox{ } \pt(C, r_2)$

\item $\pt(C,  \sum\limits_{i \in K} r_i)=  \sum\limits_{i \in K} \pt(C, r_i)       $

\item  $ \pt(\texttt{skip}, \mathbb{P}(\phi  )) = \mathbb{P}(  \phi )$

\item  $ \pt(X\leftarrow E , \mathbb{P}(\phi  )) = \mathbb{P}( \phi[X/E] )$ 

\item $ \pt(   X\xleftarrow[]{\$}  \{a_1: k_1, \ldots ,a_n:k_n\} ,       \mathbb{P}(\phi  )) =$\\ 
$ a_1    \mathbb{P} ( \phi[X/k_1] \wedge \neg  \phi[X/k_2] \wedge \ldots \neg \phi[X/k_n]    )  + $\\ 
$a_2    \mathbb{P} ( \neg \phi[X/k_1] \wedge    \phi[X/k_2] \wedge \neg \phi[X/k_3] \wedge  \ldots  \wedge \neg \phi[X/k_n]    ) +  \ldots + $
$a_n    \mathbb{P} ( \neg \phi[X/k_1] \wedge     \ldots \wedge \neg \phi[X/k_{n-1}] \wedge    \phi[X/k_n]  ) + $\\ 
$(a_1 +a_2) \mathbb{P} ( \phi[X/k_1] \wedge   \phi[X/k_2] \wedge  \neg   \phi[X/k_3] \wedge \ldots \wedge \neg \phi[X/k_n]    ) + \ldots +  (a_1+ \ldots + a_n)  \mathbb{P} ( \phi[X/k_1] \wedge    \ldots \wedge \phi[X/k_n]    )  $

\item $\pt(U[\overline{q}] ,  \mathbb{P}(\phi  ) ) =  \mathbb{P}( [U_{\overline{q}}]\phi  )$

\item $\pt(X\leftleftarrows q_j ,   \mathbb{P}(\phi  ))=    (\top \Rightarrow  P^0_j )    \mathbb{P}(    [\mathfrak{P}^0_j ]\phi [X/0]       )  +   (\top \Rightarrow  P^1_j )      \mathbb{P}(    [\mathfrak{P}^1_j ]   \phi [X/1]) + \mathbb{P}([P_j^0]\phi[X/0]\wedge [P_j^1]\phi[X/1])$

\item $\pt(C_1;C_2 ,\mathbb{P}(\phi  ) ) = \pt(C_1, \pt(C_2, \mathbb{P}(\phi  ))) $

\item  $ \pt(  \texttt{if}\ B\ \texttt{then}\ C_{1}\ \texttt{else}\ C_{2} ,  \mathbb{P}(\phi  )  ) =  \pt(C_1,  \mathbb{P}(\phi  )  )/B   +  \pt(C_2,  \mathbb{P}(\phi  )  )/(\neg B)  $  

\item $ \pt( \texttt{while} \mbox{ } B \mbox{ } \texttt{do} \mbox{ } C ,   \mathbb{P}(\phi  ))   =  \displaystyle\sum_{i = 0}^{\infty}  T_i $, in which $T_i$ is defined via the following procedure.  We use the following abbreviation.

\begin{itemize}

\item $\wpd(i) $ is short for $  \neg \wpd(C^0 , \neg B ) \wedge \ldots \wedge    \neg \wpd(C^{i-1} , \neg B ) \wedge   \wpd(C^{i} , \neg B ) $

\item $\wpd(\infty) $ is short for $ \displaystyle \bigwedge_{i=0}^{ \infty}  \neg  \wpd(C^{i} , \neg B ) $

\item $WL$ is short for $\texttt{while} \mbox{ } B \mbox{ } \texttt{do} \mbox{ } C $

\item $IF$ is short for $\texttt{if}\ B\ \texttt{then}\ C \ \texttt{else} \mbox{ }  \texttt{skip}$

\item   $SUM$ is short for $\displaystyle\sum_{i = 0}^{\infty} (  \mathbb{P}(    \wpd(i)    )      ( \pt( (  IF)^i  ,\mathbb{P}(\phi  )  ) / (    \wpd(i)   )) )$

\end{itemize}

For any probabilistic cq-state $\Theta$, we let\\
 $\Theta_0 = \downarrow_{\wpd(\infty)}(\Theta)$, $\Theta_{i+1}=\downarrow_{\wpd(\infty)}(  [\![C]\!] \Theta_i) $.
 
 \begin{itemize}
\item   Let $T_0 = SUM$.
\item Let $T_1$ be the unique real expression  such that 
 $$[\![  T_1]\!]_{\Theta } = [\![   \mathbb{P}(  \wpd(\infty)  )     ]\!]_{\Theta }  [\![ SUM]\!]_{ [\![ C]\!] \Theta_0 }      $$ for all probabilistic cq-state $\Theta$.
Equivalently,  $$T_1  =   \mathbb{P}(  \wpd(\infty)  )     (       \pt(C,SUM) /\wpd(\infty)        )   .$$
\item Let $T_2$   be the unique real expression  such that 
$$[\![  T_2]\!]_{\Theta } = [\![   \mathbb{P}(  \wpd(\infty)  )     ]\!]_{\Theta }     \mathbb{P}(  \wpd(\infty)  )     ]\!]_{ [\![ C]\!] \Theta_0 }          [\![ SUM]\!]_{ [\![ C]\!] \Theta_1 }      $$
 for all probabilistic states $\Theta$. Equivalently, $T_2=$
 $$  \mathbb{P}(  \wpd(\infty)  )         (       \pt(C,      \mathbb{P}(  \wpd(\infty)  )    ) /\wpd(\infty)        ) (            \pt(C,             \pt(C,SUM) /\wpd(\infty)          ) /\wpd(\infty)                         )   $$
 
 \item  In general, we let $T_i$   be the unique real expression such that $[\![  T_i]\!]_{\Theta} =$
 $$  [\![   \mathbb{P}(  \wpd(\infty)  )     ]\!]_{\Theta}      [\![\mathbb{P}(  \wpd(\infty)  )     ]\!]_{ [\![ C]\!] \Theta_0 }   \ldots $$
 $$  [\![\mathbb{P}(  \wpd(\infty)  )     ]\!]_{ [\![ C]\!] \Theta_{i-2} }  [\![ SUM]\!]_{ [\![ C]\!] \Theta_{i-1} }  $$
 for all probabilistic cq-state $\Theta$. Equivalently, let $f_{C,B}$ be the function that maps a real expression $r$ to $f_{C,B}(r) = \pt(C,r) / \wpd(\infty)  $. Then $T_i$ can be characterized by the following formula:
 
 $$T_i = f_{C,B}^i(SUM)  \displaystyle\prod_{j = 0}^{i-1}     f_{C,B}^j ( \mathbb{P}(  \wpd(\infty)  ) ).  $$
 
 Therefore, $$ \pt( \texttt{while} \mbox{ } B \mbox{ } \texttt{do} \mbox{ } C ,   \mathbb{P}(\phi  ))   =  \displaystyle\sum_{i = 0}^{\infty}   f^i(SUM)  \displaystyle\prod_{j = 0}^{i-1}     f^j ( \mathbb{P}(  \wpd(\infty)  ) ).  $$
 
 \end{itemize}

\item $\pt(\texttt{skip},  \phi \Rightarrow Q ) = \phi \Rightarrow Q $

\item $\pt(X\leftarrow E, \phi \Rightarrow Q  )   =    \phi [X /E] \Rightarrow Q $

\item $\pt( X\xleftarrow[]{\$}  \{a_1: k_1, \ldots ,a_n:k_n\} ,  \phi \Rightarrow Q )  =$\\ 
$ a_1( \phi[X/k_1]  \Rightarrow Q    ) + \ldots + a_n( \phi[X/k_n]  \Rightarrow Q    ) $

\item $\pt( U[\overline{q}],  \phi \Rightarrow  Q  )  =  [U_{\overline{q}}]\phi \Rightarrow U^{\dagger}_{\overline{q}}Q U_{\overline{q}}$

\item $\pt(X\leftleftarrows q_j , \phi \Rightarrow Q  ) =   (\phi[X/0 ]  \Rightarrow    P^0_j  Q  P^0_j        ) +    (\phi[X/1 ]  \Rightarrow      P^1_j  Q  P^1_j     ) $

\item $\pt(C_1;C_2, \phi \Rightarrow Q ) = \pt(C_1, \pt( C_2, \phi \Rightarrow Q)  )$

\item $\pt( \texttt{if}\ B \   \texttt{then}\ C_{1}\ \texttt{else}\ C_{2} , \phi \Rightarrow Q  )= $ \\
$\pt(C_1, \phi \Rightarrow Q)/B + \pt( C_2, \phi \Rightarrow Q)/ \neg B  $

\item $\pt(WHILE, \phi \Rightarrow Q)= \pt(WHILE, \mathbb{P}(\psi) ) [  \mathbb{P}(\psi) / (\phi \Rightarrow Q)   ]$

\end{enumerate}

\end{definition}

\noindent
Note that   $\pt(WHILE, \mathbb{P}(\psi) ) [  \mathbb{P}(\psi) / (\phi \Rightarrow Q)   ]$ is an informal notation. It means that $\pt(WHILE, \phi \Rightarrow Q)$ is obtained by replacing the appearance of $\mathbb{P}(\psi)$ in $\pt(WHILE, \mathbb{P}(\psi) )$ by $ (\phi \Rightarrow Q)$. The following lemma shows that our definition of weakest preterm captures the intuition of weakest preterm.

\begin{lemma} \label{preterm lemma}
Given an arbitrary probabilistic cq-state $\Theta $, an interpretation $I$, a command $C$, and a real expression $r$ (excluding cq conditionals),   

   $$ [\![  pt  ( C,  r ) ]\!]^I_{\Theta } =  [\![ r ]\!]^I_{ [\![ C]\!] \Theta} $$
\end{lemma}

\begin{proof}
Please refer to Appendix \ref{proof:lemma-pt}.
\end{proof}

The major  contribution  of this paper is extending the concept of weakest preterm from probabilistic Hoare logic to quanutum Hoare logic. With weakest preterm at hand, we now straightforwardly define the weakest precondition of probabilistic assertion and develop a proof system of QHL with probabilistic assertion. 

\begin{definition}[Weakest precondition of probabilistic assertion]\label{def.weakestpreconditionprobabilistic}

\begin{enumerate}
\item $\WP(C, r_1 \mbox{ }  rop \mbox{ } r_2  ) = \pt(C, r_1) \mbox{ }  rop \mbox{ }  \pt(C, r_2)$
\item $  \WP(C, \neg \Phi) = \neg \WP(C, \Phi) $
\item $\WP(C,\Phi_1 \wedge \Phi_2) = \WP(C,\Phi_1) \wedge \WP(C,\Phi_2) $ 
\end{enumerate}

\end{definition}

\begin{theorem}\label{wp for PA}
 
   $\mu \models^I \WP(C, \Phi)$ iff $ [\![C]\!]\mu \models^I \Phi$.

\end{theorem}

\noindent
We provide the proof in Appendix \ref{proof:wp-PA}. The theorem mentioned above inspires a definition for a QHL proof system with respect to probabilistic formulas in a uniform manner. The proposed definition is $\vdash \{\WP(C, \Phi)\} C \{\Phi\}$, which applies to every command $C$ and probabilistic formula $\Phi$. We will follow this construction with some minor modifications.

\subsection{Proof system with probabilistic assertions}

\begin{definition}[Proof system with probabilistic assertions] The proof system of our QHL with probabilistic assertions consists of the following inference rules:

 \begin{center}
        \begin{tabular}{ll}
         $SKIP:$ & $\frac{}{\vdash\{  \Phi    \}\texttt{skip}\{  \Phi  \}}$\\
         $AS:$ & $ \frac{}{\vdash \{      \WP(X \leftarrow E, \mbox{ }\Phi)   \}  X \leftarrow E \{  \Phi \} }$\\
       
         $PAS:$ & $   \frac{}{\vdash \{  \WP(X\xleftarrow{\$}\{a_1: r_{1},...,a_n: r_{n} \},\mbox{ } \Phi)    \}  X\xleftarrow{\$}\{a_1: r_{1},...,a_n: r_{n} \}   \{    \Phi    \} }$\\  
               
         $UNITARY:$        & $ \frac{}{\vdash \{      \WP(U[\overline{q}], \mbox{ }\Phi)   \} U[\overline{q}] \{  \Phi \} }$\\

         $MEASURE:$        & $ \frac{}{\vdash \{      \WP(X\leftleftarrows q_j, \mbox{ }\Phi)   \} X\leftleftarrows q_j \{  \Phi \} }$\\   
         
         \\
         $SEQ:$ & $ \frac{\vdash\{\Phi\}C_{1}\{\Phi_{1}\}\quad \vdash\{\Phi_{1}\}C_{2}\{\Phi_{2}\}} {\vdash\{\Phi\}C_{1};C_{2}\{\Phi_{2}\}} $\\
     
         $IF:$ & $\frac{}{\vdash\{  \WP( \texttt{if}\ B\ \texttt{then}\ C_{1}\ \texttt{else}\ C_{2}  , \mbox{ } \Phi)    \}  \texttt{if}\ B\ \texttt{then}\ C_{1}\ \texttt{else}\ C_{2}   \{  \Phi \}}$\\

         $WHILE:$ & $\frac{}{\vdash\{   \WP( \texttt{while} \mbox{ }B \mbox{ } \texttt{do} \mbox{ }C , \mbox{ } \Phi)   \}  \texttt{while} \mbox{ }B \mbox{ } \texttt{do} \mbox{ }C   \{ \Phi \}}$,\\
                 
                 \\
          $CONS:$ & $ \frac{\models\Phi'\rightarrow\Phi\quad  \vdash\{\Phi\}C\{\Psi\}\quad \models\Psi\rightarrow\Psi'}  {\vdash\{\Phi'\}C\{\Psi'\}}$

        \end{tabular}
    \end{center}

\end{definition}

With the rule (CONS) in our proof system, we can treat $\WP(C,\Phi)$ as a semantic notion: if $\models \WP(C,\Phi) \leftrightarrow \Psi$, then $\Psi$ is conceived as the weakest precondition of $\Phi$ with command $C$.

\begin{lemma}\label{wp condition prob}
If $\models  \{ \Phi \} C \{ \Psi \} $, then $\models \Phi \rightarrow \WP(C, \Psi)$.
\end{lemma}

\begin{proof}
Let $\mu$ be a probabilistic state such that $\mu \models \Phi$. Then $ [\![ C]\!] \mu \models \Psi $. Then by Theorem \ref{wp for PA} we know $\mu \models \WP(C,\Psi)$.
\end{proof}

\begin{lemma}
   For an arbitrary command $C$ and an arbitrary probabilistic formula $\Phi$, we have $\vdash \{\WP(C,\Phi)\}C\{\Phi\}$.
\end{lemma}

\begin{proof}
    Please refer to Appendix \ref{proof:WP-syntactic}.
\end{proof}

\begin{theorem}

The QHL proof system with probabilistic assertions is sound and complete.

\end{theorem}

\begin{proof}

Soundness: the soundness of SKIP, AS, PAS, IF, WHILE follows from Theorem \ref{wp for PA}. The soundness of SEQ and CON can be proved by simple deduction.

Completeness: Assume $\models  \{ \Phi \} C \{ \Psi \} $, then $\models \Phi \rightarrow \WP(C, \Psi)$ by the lemma \ref{wp condition prob}. Then by $\vdash \{\WP(C,\Psi) \}C \{ \Psi \} $ and CONS we know $ \vdash  \{\Phi \}C \{ \Psi \} $.

\end{proof}

To showcase the expressive and inference power of our proposed QHL, we formally verify two widely discussed quantum algorithms - Deutsch's algorithm and quantum teleportation.  Due to page limitations, details are shown in Appendix D.

\section{Conclusion and future work}
\label{sec:conclusion}

Inspired by the recent work on proving the relative completeness of satisfaction-based probabilistic Hoare logic \cite{SunSBLC24}, in this paper we prove the relative completeness of satisfaction-based quantum Hoare logic. The completeness is proved in two steps. First, we establish a semantics and proof system of Hoare triples with quantum programs and deterministic assertions. Then, utilizing the weakest precondition of deterministic assertion, we construct the weakest preterm calculus of probabilistic expressions. The relative completeness of QHL is then obtained as a consequence of the weakest preterm calculus.


It could be a good future project to extend this result to quantum relational Hoare logic. However, we anticipate that achieving completeness for such a logic could be quite challenging. Another possible direction is to implement our QHL and use the implementation to verify numerous quantum algorithms.

\newpage
\bibliographystyle{splncs04}
 
\bibliography{literature}

\appendix

\section{Basics of quantum computation} \label{sec:basics}
We will introduce the basic ideas of quantum computation. For more
information on quantum computation, see~\cite{nielsen2001quantum, Sel2004-qpl}. The basic concepts of quantum
computation are \emph{states} and \emph{operations} which act on
states. Let $\C$ be the field of complex numbers. We write $\C ^ n$ for
the space of $n$-dimensional column vectors.

\subsubsection{$n$-qubit states}

\begin{definition}
An $n$-qubit state is a unit vector in $\C ^ {2 ^ n}$. 
\end{definition}

In other words, the state space is the set of $2^n$-dimensional
complex unit vectors. For example, the space of the $1$-qubit states is
the unit sphere in the $2$-dimensional complex vector space. 

For $\C ^ {2^n}$, we use the ``ket-binary'' notation to denote the
standard basis. For example, when $n=2$:
\[
  \ket{00} = \begin{bmatrix}
    1  \\
    0  \\
    0  \\
    0 
  \end{bmatrix}, \quad
  \ket{01} = \begin{bmatrix}
    0  \\
    1  \\
    0  \\
    0 
  \end{bmatrix}, \quad
  \ket{10} = \begin{bmatrix}
    0  \\
    0  \\
    1  \\
    0 
  \end{bmatrix}, \quad
  \ket{11} = \begin{bmatrix}
    0  \\
    0  \\
    0  \\
    1 
  \end{bmatrix}.
\]

In general $\ket{b_nb_{n-1}...b_0} = e_k$, where $ k = \sum_{j=0}^{n}b_j2^{j}$, and $e_k$ is the $k$-th standard basis vector.
Note that $b_nb_{n-1}...b_0$ is the binary representation of $k$, and
that the standard basis starts from $e_0$. In the following, we will
identify each natural number index with its binary representation. One
convenience of this notation is
\[\ket{a_ma_{m-1}...a_0}\otimes \ket{b_{n}b_{n-1}...b_0} = \ket{a_{m}a_{m-1}...a_0b_{n}b_{n-1}...b_0}.\]
Another way to look at this is that every basis vector of
$\C^{2^{m+n}}$ is a tensor product of basis vectors of $\C^{2^{m}}$
and $\C^{2^{n}}$. We can also identify $\C^{2^{n}}$ with the $n$-fold
tensor product of $\C^2$. The ket-binary notation reflects this
identification.


There are two and only two kinds of operations that can be applied to
an $n$-qubit state --- \emph{unitary transformations} and
\emph{measurements}. Quantum algorithms are often of the form that
first apply some number of unitary transformations to a an $n$-qubit
state and then do a measurement, and based on the measurement result,
some output or conclusion is drawn.

\subsubsection{$n$-qubit unitary transformations}

A linear function $U : \C^k \rightarrow \C^m$ is called an
\emph{isometry} if for all $v,w$ in $\C^k$, we have
$\gen{Uv, Uw} = \gen{v,w}$. $U$ is called \emph{unitary} if it is an
invertible isometry. Note that $U$ is unitary if and only if
$UU^\dagger = I$ and $U^\dagger U = I$, where $U^\dagger$ denotes the
adjoint of $U$. For a matrix, the adjoint is the complex conjugate of
the transpose.

\begin{definition}
  An $n$-qubit unitary transformation is a unitary transformation on $\C^{2^n}$
\end{definition}

We also refer to an $n$-qubit unitary transformation as an
\emph{$n$-qubit operator}. The unitary condition is
needed to preserve the norm of a state so that applying a unitary
transformation to a state produces a state.

\subsubsection{Measurement}

The other operation is the measurement. A measurement is like a
finite-valued random variable with ``side-effect'' which modifies the
qubit state being measured. 

\begin{definition}
  Given an $n$-qubit state
  $\theta = \sum_{j=00...0}^{11...1} \alpha_j \ket{j}$, a $1$-qubit measurement
  acting on the $k$-th qubit outputs $b$ and alters the state to
  \[
  \sum_{j=00...0}^{11...1}\sum_{l=00...0}^{11...1} \alpha'_{jbl}
  \ket{jbl}
  \]with a probability
  $p = \Tr( (I_{2^j} \otimes \ket{b}\bra{b} \otimes I_{2^l}) \theta    \theta^{\dagger})$, where $b = 0,1$, $I_{2^x}$ is the $2^x$-dimensional identity operator (i.e. identity operator on x-qubits), 
  $\alpha'_{jbl} = \alpha_{jbl}/\sqrt{p}$, $j$ is of length $k-1$, and $l$ is
  of length $n-k$, and $jbl$ is a concatenation of binary strings $j$,
  $b$, and $l$.
\end{definition}

Later we will use $\Tr( ( |b\rangle \langle b| \otimes I_{-k})   \theta    \theta^{\dagger})$ to denote $\Tr( (I_{2^j} \otimes \ket{b}\bra{b} \otimes I_{2^l}) \theta    \theta^{\dagger})$.

\section{Proofs in Section \ref{sec:deterministic}}

In this article, all proofs which involve no quantum operation are essentially the same as their analogues in \cite{SunSBLC24}. 

\begin{theorem*}[Soundness]\label{proof:soundness-d}
 The proof system QHL$_d$ is sound, $i.e.$ for all deterministic formula $\phi$ and $\psi$ and command $C$, $\vdash\{\phi\}C\{\psi\}$ implies $\models\{\phi\}C\{\psi\}$. 
\end{theorem*}

\begin{proof}

We prove by structural induction on $C$.  Let $I$ be an arbitrary interpretation.

\begin{itemize}
\item (SKIP) It's trivial to see that that $\models \{\phi\}\texttt{skip}\{\phi\}$.

\item (AS) Assume $ \theta \models^I \phi[X/ E ] $. We prove it by induction on the structure of $\phi$. The only non-trivial case is when $\phi$ is of the form $[U_{\overline{q}}] \psi$. 

Assume $\theta \models [U_{\overline{q}}] \psi   [X/E]$.  Then $(\theta_c, U_{\overline{q}} \theta_q ) \models \psi [X/E]$. By the inference rule (AS), we have $\vdash \{\psi[X/E]\}X\leftarrow E\{\psi\}$.  By induction hypothesis (IH), we obtain that $\models \{\psi[X/E] \} X\leftarrow E \{ \psi\} $. Therefore, it holds that  $([\![X\leftarrow E ]\!] \theta_c, U_{\overline{q}} \theta_q ) \models \psi  $ and  $( [\![X\leftarrow E ]\!] \theta_c,  \theta_q ) \models  [U_{\overline{q}}] \psi$ i.e. $[\![ X\leftarrow E]\!]\widehat{\theta}\models [U_{\overline{q}}]\psi$.

\item (PAS) Assume $ \theta  \models^I \phi[X/k_1] \wedge \ldots \wedge  \phi[X/k_n] $. This means that $\phi$ is true if the variable $X$ is assigned to the any of $\{k_1, \ldots, k_n\}$ and all other values are assigned to a value according to $\theta_c$. Let $\Theta = [\![   X\xleftarrow{\$}\{a_1: k_{1},...,a_n: k_{n}\}     ]\!]( \widehat{\theta} )$. Then $\Theta$ is a distribution with support $\{\theta_1, \ldots, \theta_n\}$, where $ \theta_i =  \theta [X\mapsto k_i]$ for $i\in \{1, \ldots,n\}$. Since $\theta_{i,c}$ assigns $X$ to the value $k_i$ and all other variables to the same value as $\theta_c$ and $\theta_{i,q}= \theta_q$. We know that $\theta_i \models^I \phi$. This means that $\phi$ is true on all supports of $\Theta$. Therefore, $\Theta \models^I \phi$.

\item (SEQ)  If the rule (SEQ) is used to derive $\vdash\{\phi\}C_{1};C_{2}\{\phi_{2}\}$ from $\vdash\{\phi\}C_{1}\{\phi_{1}\}$ and $\vdash\{\phi_{1}\}C_{2}\{\phi_{2}\}$, then by IH, we have $\models\{\phi\}C_{1}\{\phi_{1}\}$ and $\models\{\phi_{1}\}C_{2}\{\phi_{2}\}$. Assume $\theta  \in [\![\phi]\!]^I$. Let $\theta'$ be an arbitrary pure cq-state such that $\theta' \in sp(   [\![C_{1}; C_{2}]\!](\widehat{\theta})  )$. By $[\![C_{1}; C_{2}]\!](\widehat{\theta} )  =    [\![C_{2}]\!]([\![C_{1}]\!](  \widehat{\theta} ))   $, we know that there is a state $\theta_1 \in sp(   [\![C_{1}]\!]( \widehat{\theta})  )$ such that $\theta' \in sp (   [\![C_{2}]\!](\widehat{\theta_1})   )$. Now by  $\models  \{ \phi \} C_{1} \{ \phi_{1} \}$   we know that $ [\![C_{1}]\!](\widehat{\theta})  \models^I \phi_1$ and hence $\theta_1 \models^I \phi_1$. By $\models \{ \phi_{1} \}C_{2} \{ \phi_{2} \}$ we know that $ [\![C_{2}]\!]( \widehat{\theta_1} )  \models^I \phi_2$ and hence  $\theta'\models^I \phi_2$.

\item (IF)   Assume $ \theta  \in [\![\phi]\!]^I$, $\vdash\{\phi\wedge B\}C_{1}\{\psi\}$ and $\vdash\{\phi\wedge \neg B\}C_{2}\{\psi\}$. By induction hypothesis we know that $ \models  \{\phi\wedge B\}C_{1}\{\psi\}$ and $\models  \{\phi\wedge \neg B\}C_{2}\{\psi\}$. Let $\theta'$ be an arbitrary state which belongs to $sp([\![\texttt{if}\ B\ \texttt{then}\ C_{1}\ \texttt{else}\ C_{2}]\!]( \widehat{\theta} ))$.

Since $\theta$ is a pure cq-state, we know that either $ \theta \in  [\![B]\!] $ or $ \theta \in  [\![\neg B]\!] $.

\begin{itemize}

\item If $ \theta \in  [\![B]\!] $, then  $[\![\texttt{if}\ B\ \texttt{then}\ C_{1}\ \texttt{else}\ C_{2}]\!](  \widehat{\theta})=[\![C_{1}]\!](\downarrow_{B}( \widehat{\theta}))  =   [\![C_{1}]\!]( \widehat{\theta})  $. Hence $\theta'\in  sp ( [\![C_{1}]\!](  \widehat{\theta} ) )$. From $\theta  \in [\![\phi]\!]^I$ and $ \theta \in  [\![B]\!] $ we know that $\theta \models^I  \phi \wedge B$. Now by $ \models  \{\phi\wedge B\}C_{1}\{\psi\}$ we deduce that $[\![C_{1}]\!](\widehat{\theta} ) \models^I \psi$. Therefore, $\theta' \models^I  \psi$.

\item If $ \theta \in  [\![ \neg B]\!] $, then  $[\![\texttt{if}\ B\ \texttt{then}\ C_{1}\ \texttt{else}\ C_{2}]\!]( \widehat{\theta})=[\![C_{2}]\!](\downarrow_{ \neg B}( \widehat{\theta}))  =   [\![C_{2}]\!]( \widehat{\theta})  $. Hence $\theta'\in sp([\![C_{2}]\!](\widehat{\theta} ) )$. From $\theta  \in [\![\phi]\!]^I$ and $ \theta \in  [\![\neg B]\!] $ we know that $\theta \models^I  \phi \wedge \neg B$. Now by $ \models  \{\phi\wedge \neg B\}C_{2}\{\psi\}$ we deduce that $[\![C_{2}]\!]( \widehat{\theta}) \models^I \psi$. Therefore, $\theta' \models^I  \psi$.
\end{itemize}

\item (WHILE) Assume that $\vdash \{\phi  \wedge B\}   C   \{\phi\}$, By IH, it holds that $\models \{\phi  \wedge B\}   C   \{\phi\} $. Let $\theta$ be an arbitrary pure cq-state and $I$ be an arbitrary interpretation such that $\theta\models^I \phi$. We remind the readers that $ [\![   \texttt{while} \mbox{ } B \mbox{ } \texttt{do} \mbox{ } C  ]\!] (\theta) =  \displaystyle\sum_{i = 0}^{\infty} \downarrow_{\neg B}  (  (  [\![   C   ]\!] \circ \downarrow_{ B})^i (\theta ) )$. We want to show that for each natural number $k$, it holds that  $     (  [\![   C   ]\!] \circ \downarrow_{ B})^k (\theta ) \models^I \phi  $.

We prove it by induction on $k$:

\begin{itemize}
\item If $k=0$, then $([\![C]\!] \circ \downarrow_{ B})^k ( \theta )  =       \widehat{\theta}   $.  Hence $  (  [\![   C   ]\!] \circ \downarrow_{ B})^k (\theta )  \models^I  \phi$.

\item Assume that $     (  [\![   C   ]\!] \circ \downarrow_{ B})^k (\theta) \models^I \phi  $.
Then $     (  [\![   C   ]\!] \circ \downarrow_{ B})^{k+1} (\theta)  =      [\![   C   ]\!] \circ \downarrow_{ B}  ( (  [\![   C   ]\!] \circ \downarrow_{ B})^{k} (  \theta )   )$. For arbitrary $\theta' \in sp (   (  [\![   C   ]\!] \circ \downarrow_{ B})^k (\theta )   )$, we have $\theta' \models^I \phi$. 
\begin{itemize}
    \item If $\theta' \models \neg B$, then $sp([\![   C   ]\!] \circ \downarrow_{ B} ( \theta')) =\emptyset$. Then $[\![   C   ]\!] \circ \downarrow_{ B} ( \theta') \models^I \phi$ vacuously.
    \item If $\theta' \models  B$, then $[\![   C   ]\!] \circ \downarrow_{ B} ( \theta')= [\![   C   ]\!]   ( \theta')$. Then by $\models \{\phi  \wedge B\}   C   \{\phi\} $ we know $[\![   C   ]\!]  ( \theta') \models^I \phi$. 
\end{itemize}

\end{itemize}

Therefore, we know $[\![   C   ]\!] \circ \downarrow_{ B} (\theta') \models^I \phi$ for all  $\theta' \in sp (   (  [\![   C   ]\!] \circ \downarrow_{ B})^k (\theta )   )$. This implies that $ (  [\![   C   ]\!] \circ \downarrow_{ B})^{k+1} (\theta ) \models^I \phi$.
According to the above claim, we infer that $\downarrow_{\neg B}  (  (  [\![   C   ]\!] \circ \downarrow_{ B})^i ( \theta ) ) \models^I \phi \wedge \neg B$ for each $i\in\mathbb{N}$. 	
Hence $\displaystyle\sum_{i = 0}^{\infty} \downarrow_{\neg B}  (  (  [\![   C   ]\!] \circ \downarrow_{ B})^i ( \theta ) ) \models^I \phi  \wedge \neg B$. This proves $  [\![   \texttt{while} \mbox{ } B \mbox{ } \texttt{do} \mbox{ } C  ]\!] (  \theta ) \models^I  \phi  \wedge \neg B $.

\item (UNITARY) If $\theta \models [U_{\overline{q}}] \phi$, then $ [U_{\overline{q}}] \theta \models \phi $. That is, $[\![ U[ \overline{q}] ]\!] \theta \models \phi$.

\item (MEASURE) Assume $\theta \models ( [\mathfrak{P}^0_j ]\phi[X/0] \vee P^1_j )\wedge ( [\mathfrak{P}^1_j ]\phi[X/1] \vee P^0_j)$. We prove it by case analysis:

\begin{itemize}
\item If $\theta \models P^1_j$, then $\theta \not\models P^0_j$. Hence $\theta \models  [\mathfrak{P}^1_j ] \phi[X/1]$.  Note that in this case $ [\![  X    \leftleftarrows q_j ]\!] \theta = \widehat{\theta'} $ where $\theta'_c = \theta_c [X \mapsto 1]$  and $\theta'_q=\frac{(P^1 \otimes I_{-j}) \theta_q      }{  \sqrt{p_1}   } $, $p_1= \Tr( (P^1 \otimes I_{-j})   \theta_q  \theta_q^{\dagger}          )  $. From  $\theta \models  [\mathfrak{P}^1_j ] \phi[X/1]$, we obtain that $( \theta_c , \theta'_q) \models  \phi[X/1]$, which further implies that $( \theta_c [X \mapsto 1], \theta'_q) \models  \phi $. That is $ \theta' \models  \phi$.

\item The case for $\theta \models P^0_j$ can be proved analogously.

\item If $\theta \models \neg P^0_j \wedge \neg P^1_j $, then $\theta \models  [\mathfrak{P}^0_j ] \phi[X/0]  \wedge [\mathfrak{P}^1_j ] \phi[X/1]$. In this case $ [\![  X  \leftleftarrows q_j ]\!] \theta =  \Theta$ where $\Theta = p_0 \theta_0 + p_1 \theta_1$, $p_i = \Tr( (P^i \otimes I_{-j})   \theta_q   \theta_q^{\dagger}        ) >0 $, $\theta_{i,c} =  \theta_c [X \mapsto i]$, $\theta_{i, q}= \frac{(P^i \otimes I_{-j}) \theta_q      }{  \sqrt{p_i}   } $. From $\theta \models [\mathfrak{P}^0_j ] \phi[X/0] $ we know $(\theta_c, \theta_{0,q}  ) \models  \phi[X/0] $, which implies that $(\theta_{0,c}, \theta_{0,q} ) \models \phi $.  From $\theta \models [\mathfrak{P}^1_j ] \phi[X/1] $ we know $(\theta_c, \theta_{1,q}  ) \models  \phi[X/1] $, which implies that $(\theta_{1,c}, \theta_{1,q} ) \models \phi $. Therefore, $\Theta \models \phi$.

\end{itemize}

\item (CONS) Assume $\models \phi'\rightarrow \phi$, $\vdash \{\phi\} C \{\psi\} $ and $\models \psi \rightarrow \psi'$. By induction hypothesis we obtain that $\models \{\phi\} C \{\psi\}$. Let $\theta$ be a state such that $\theta \models^I  \phi'$. Then by $\models \phi'\rightarrow \phi$ we know that $\theta \models^I  \phi$. By $\models \{\phi\} C \{\psi\}$ we know that $ [\![C]\!] (\widehat{\theta}) \models^I \psi $. Hence $\theta' \models^I \psi$ for all $\theta'$ which belongs to $sp([\![C]\!] (\widehat{\theta}))$. Now by  $\models \psi \rightarrow \psi'$ we know that $\theta' \models^I \psi'$.

\end{itemize}    

\end{proof}

\begin{theorem*} \label{proof:wpd}

$\Theta \models^I \wpd(C,\phi)$ iff $[\![ C]\!] (\Theta) \models^I \phi$.

\end{theorem*}

\begin{proof}
It is sufficient to prove that $\theta \models^I \wpd(C,\phi)$ iff $[\![ C]\!] (\widehat{\theta}) \models^I \phi $.

\begin{itemize}

\item (skip) $\theta \models^I \wpd(skip,\phi)$ iff $ \theta \models^I \phi$ iff $[\![ skip ]\!] (\widehat{\theta}) \models^I \phi $.

\item (AS)  $\theta \models^I \wpd(X \leftarrow E,\phi)$ iff $\theta \models^I \phi [X/E]$. Then we will show that  $[\![ X \leftarrow E ]\!] (\widehat{\theta}) \models \phi $.

Assume that $\theta\models^I \phi[X/E]$. By Definition~\ref{def.semanticsofcommand}, $[\![X\leftarrow E]\!](\widehat{\theta})=(\theta_c,\theta'_q)$ where $\theta'_q=\theta_q$ and $\theta'_c=\theta_c[X\mapsto \llbracket E\rrbracket^I\theta_c]$. 

Do induction on the structure of $\phi$: 

\begin{itemize}
   
\item  When $\phi$ is $\top$ or $\bot$, the proof is trivial.

\item When $\phi=P^i_j$, $\theta\models^I P^i_j[X/E]$ iff $\theta\models^I P^i_j$ iff $Tr(P^i Tr_{-j}(\theta_q\theta_q^{\dagger}))=1$ iff $(\theta'_c,\theta_q)\models P_j^i$ iff $\llbracket X\leftarrow E\rrbracket \widehat{\theta}\models P_j^i$. 

\item When $\phi=e_1\ rop\ e_2$, $\theta\models^I (e_1\ rop\ e_2)[X/E]$ iff $\theta\models^I e_1[X/E]\ rop\ e_2[X/E]$ iff $\theta_c\models^I e_1[X/E]\ rop\ e_2[X/E]$.

      Induction on the structure of $e_1$ (similarly to $e_2$), we will prove $\llbracket e_1[X/E]\rrbracket^I\theta_c=\llbracket e_1\rrbracket^I_{\theta_c[X\mapsto \llbracket E\rrbracket^I\theta_c]}$. 

      \begin{itemize}
          \item If $e_1=n$ for some integer $n$, $\llbracket n[X/E]\rrbracket^I\theta_c=n=\llbracket n\rrbracket^I\theta_c[X\mapsto \llbracket E\rrbracket\theta_c]$.
          \item If $e_1=X$, $\llbracket X[X/E]\rrbracket^I\theta_c=\llbracket E\rrbracket\theta_c =\llbracket X\rrbracket^I\theta_c[X\mapsto \llbracket E\rrbracket^I\theta_c]$.
          \item If $e_1=x$, $\llbracket x[X/E]\rrbracket^I\theta_c=\llbracket x\rrbracket\theta_c =\llbracket x\rrbracket^I\theta_c[X\mapsto \llbracket E\rrbracket^I\theta_c]$.
          \item If $e_1=e_3\ aop\ e_4$, $\llbracket (e_3\ aop\ e_4)[X/E]\rrbracket^I\theta_c=\llbracket e_3[X/E]\rrbracket^I\theta_c\ aop\ \llbracket e_4[X/E]\rrbracket^I\theta_c$. By IH, it equals $\llbracket e_3^I\theta_c[X\mapsto \llbracket E\rrbracket\theta_c]\ aop\ e_4^I\theta_c[X\mapsto \llbracket E\rrbracket\theta_c]=\llbracket e_3\ aop\ e_4\rrbracket^I\theta_c[X\mapsto \llbracket E\rrbracket\theta_c]$. 
      \end{itemize}

      So we have $\theta_c\models^I e_1[X/E]\ rop\ e_2[X/E]$ iff $\theta_c[X\mapsto \llbracket E\rrbracket\theta_c]= e_1\ rop\ e_2$ iff $\llbracket X\leftarrow E\rrbracket (\widehat{\theta})\models^I e_1\ rop\ e_2$.

\item When $\phi=\neg\psi$, $\theta\models^I\neg\psi[X/E]$ iff $\theta\not\models^I\psi[X/E]$. By IH, $\llbracket X\leftarrow E\rrbracket (\widehat{\theta})\not\models \psi$. So $\llbracket X\leftarrow E\rrbracket (\widehat{\theta})\models \neg\psi$.

\item When $\phi=\phi_1\wedge\phi_2$, $\theta\models (\phi_1\wedge \phi_2)[X/E]$ iff $\theta\models^I\phi_1[X/E]\wedge\phi_2[X/E]$. By IH, we have $\llbracket X\leftarrow E\rrbracket^I(\widehat{\theta})\models\phi_1$ and $\llbracket X\leftarrow E\rrbracket (\widehat{\theta})\models^I\phi_2$ iff $\llbracket X\leftarrow E\rrbracket(\widehat{\theta})\models^I\phi_1\wedge\phi_2$. 

\item When $\phi=\forall x\psi$, $\theta\models^I\forall x\psi[X/E]$ iff for all integers $n$ and $I'=I[x\mapsto n]$, $\theta\models^{I'}\psi[X/E]$. By IH, $\llbracket X\leftarrow E\rrbracket (\widehat{\theta})\models^{I'}\psi$ for each $I'$ with $I'=I[x\mapsto n]$ iff $\llbracket X\leftarrow E\rrbracket (\widehat{\theta})\models \forall x\psi$.

\item When $\phi=[U_{\overline{q}}]\psi$, $\theta\models^I [U_{\overline{q}}]\psi[X/E]$ iff $\llbracket U_{\overline{q}}\rrbracket(\widehat{\theta})\models^I\psi[X/E]$ iff $(\theta_c,U_{\overline{q}}\theta U^\dagger_{\overline{q}})\models^I\psi[X/E]$. By IH, $\llbracket X\leftarrow E\rrbracket \llbracket U_{\overline{q}}\rrbracket (\widehat{\theta})\models^I\psi$ iff $(\theta_c[X\mapsto \llbracket E\rrbracket\theta_c],U_{\overline{q}} \theta_q U_{\overline{q}}^\dagger)\models^I\psi$ iff $\llbracket X\leftarrow E\rrbracket(\widehat{\theta})\models [U_{\overline{q}}]\psi$. 

\item When $\phi=[\mathcal{B}^i_j]\psi$, $\theta\models^I [\mathcal{B}^i_j]\psi[X/E]$ iff $\theta'\models^I \psi[X/E]$ iff $\theta'\models^\psi[X/E]$ where $\theta_c'=\theta_c$ and $\theta'_q=\frac{(P_j^1\otimes I_{-j})\theta_q}{\sqrt{p_i}}$. By IH, $(\llbracket X\leftarrow E\rrbracket\theta_c,\theta'_q)\models^I\psi$ iff $\llbracket X\leftarrow E\rrbracket\theta\models^I[\mathcal{B}_j^i]\psi$. 

\end{itemize}

\item (PAS) $\theta \models^I \wpd( X\xleftarrow{\$}\{a_1: k_{1},...,a_n: k_{n} \}, \phi  )$ iff $  \theta \models^I   \phi[X/k_1] \wedge \ldots \wedge  \phi[X/k_n]$. Then, it will be proved that  $  \theta \models^I   \phi[X/k_1] \wedge \ldots \wedge  \phi[X/k_n]$ iff $\llbracket X\xleftarrow{\$}\{a_1: k_{1},...,a_n: k_{n} \}\rrbracket(\widehat{\theta})\models^I\phi$. 

Do induction on the structure of $\phi$:

\begin{itemize}
    \item When $\phi$ is $\top$ or $\bot$, the proof is trivial. 
    \item When $\phi=P_j^i$, $\theta\models^I P_j^i[X/k_1]\wedge\cdots\wedge P_j^i[X/k_n]$ iff $\theta\models^I P_j^i$ iff $\llbracket X\xleftarrow{\$}\{a_1: k_{1},...,a_n: k_{n} \}\rrbracket(\widehat{\theta})\models^I P_j^i$.
    \item When $\phi= e_1\ rop\ e_2$, $\theta\models^I (e_1\ rop\ e_2)[X/k_1]\wedge\cdots\wedge (e_1\ rop\ e_2)[X/k_n]$ iff $\theta\models^I(e_1\ rop\ e_2)[X/k_1]$ and $\cdots$ and $\theta\models (e_1\ rop\ e_2)[X/k_n]$. By the proof in the case (AS) above, we have $\llbracket X\leftarrow k_1\rrbracket(\widehat{\theta})\models^I e_1\ rop\ e_2$ and $\cdots$ and $\llbracket X\leftarrow k_n\rrbracket(\widehat{\theta})\models^I e_1\ rop\ e_2$.

    ($\Rightarrow$) Let $\theta'\in sp(\llbracket X\xleftarrow{\$}\{a_1: k_{1},...,a_n: k_{n} \}\rrbracket (\widehat{\theta}))$. We know that $\theta'_q=\theta_q$ and $\theta'_c$ only differs in the value of $X$ from $\theta_c$. Suppose $\llbracket X\rrbracket^I\theta'_c=k_i$($1\leq i\leq n$). So we have $\llbracket X\leftarrow k_i\rrbracket(\widehat{\theta})=\widehat{\theta'}$ and $\widehat{\theta'}\models^I e_1\ rop\ e_2$. Since $\theta'$ is arbitrarily chosen, $\llbracket X\xleftarrow{\$}\{a_1: k_{1},...,a_n: k_{n} \}\rrbracket(\widehat{\theta})\models^I e_1\ rop\ e_2$. 

    ($\Leftarrow$) Assume $\llbracket X\xleftarrow{\$}\{a_1: k_{1},...,a_n: k_{n} \}\rrbracket(\widehat{\theta})\models^I e_1\ rop\ e_2$. It implies that for each $\theta'\in sp(\llbracket X\xleftarrow{\$}\{a_1: k_{1},...,a_n: k_{n} \}\rrbracket(\widehat{\theta}))$, $\theta'\models^I e_1\ rop\ e_2$. By Definition~\ref{def.semanticsofcommand}, $\theta'=\llbracket X\leftarrow k_i\rrbracket(\widehat{\theta})$ for some $1\leq i\leq n$. So $\llbracket X\leftarrow k_i\rrbracket(\widehat{\theta})\models e_1\ rop\ e_2$. Thus, we have $\llbracket X\leftarrow k_1\rrbracket(\widehat{\theta})\models^I e_1\ rop\ e_2$ and $\cdots$ and $\llbracket X\leftarrow k_n\rrbracket(\widehat{\theta})\models^I e_1\ rop\ e_2$.

    \item When $\phi=\neg\psi$ or $\psi_1\wedge\psi_2$ or $\forall x\psi$, the proofs are similar to those in the case of (AS), respectively.

    \item When $\phi=[U_{\overline{q}}]\psi$, $\theta\models^I [U_{\overline{q}}]\psi[X/k_1]\wedge\cdots\wedge [U_{\overline{q}}]\psi[X/k_n]$ iff $\theta\models^I [U_{\overline{q}}]\psi[X/k_1]$ and $\cdots$ and $\theta\models^I[U_{\overline{q}}]\psi[X/k_n]$ iff $(\theta_c,U_{\overline{q}}\theta_q U_{\overline{q}}^\dagger)\models^I \psi[X/k_1]$ and $\cdots$ and $(\theta_c, U_{\overline{q}}\theta_q U_{\overline{q}}^\dagger)\models^I \psi[X/k_n]$ iff, by what we proved in the case of (AS), $\llbracket X\leftarrow k_1\rrbracket(\widehat{\theta})\models [U_{\overline{q}}]\psi$ and $\cdots$ and $\llbracket X\leftarrow k_n\rrbracket(\widehat{\theta})\models [U_{\overline{q}}]\psi$. Suppose $\theta'\in sp(\llbracket X\xleftarrow{\$}\{a_1: k_{1},...,a_n: k_{n} \}\rrbracket(\widehat{\theta}))$. $\theta'_c=\llbracket X\leftarrow k_i\rrbracket(\widehat{\theta})$ for some $1\leq i\leq n$ and $\theta'_q=\theta_q$. So $\theta'=\llbracket X\leftarrow k_i\rrbracket(\theta)$.

    ($\Leftarrow$) By $\llbracket X\leftarrow k_i\rrbracket(\widehat{\theta})\models [U_{\overline{q}}]\psi$, $\theta'\models [U_{\overline{q}}]\psi$. So $\llbracket  X\xleftarrow{\$}\{a_1: k_{1},...,a_n: k_{n}\rrbracket (\widehat{\theta})\models [U_{\overline{q}}]\psi$.

    ($\Rightarrow$) By $\llbracket  X\xleftarrow{\$}\{a_1: k_{1},...,a_n: k_{n}\rrbracket (\widehat{\theta})\models [U_{\overline{q}}]\psi$, $\theta'\models [U_{\overline{q}}]\psi$, which implies that $\llbracket X\leftarrow k_i\rrbracket(\widehat{\theta})\models^I [U_{\overline{Q}}]\psi$.

     So we proved that $\theta\models^I [U_{\overline{q}}]\psi[X/k_1]\wedge\cdots\wedge [U_{\overline{q}}]\psi[X/k_n]$ iff $\llbracket X\xleftarrow{\$}\{a_1: k_{1},...,a_n: k_{n} \} \rrbracket(\widehat{\theta})\models [U_{\overline{q}}]\psi$. 

     \item When $\phi=[\mathcal{B}_j^i]\psi$, $\theta\models^I[\mathcal{B}_j^i]\psi[X/k_1]\wedge\cdots\wedge [\mathcal{B}_j^i]\psi[X/k_n]$ iff $\theta\models^I[\mathcal{B}_j^i]\psi[X/k_1]$ and $\cdots$ and $\theta\models^I[\mathcal{B}_j^i]\psi[X/k_n]$ iff $\llbracket X\leftarrow k_1\rrbracket(\widehat{\theta})\models [\mathcal{B}_j^i]\psi$ and $\cdots$ and $\llbracket X\leftarrow k_n\rrbracket(\widehat{\theta})\models [\mathcal{B}_j^i]\psi$. Suppose that $\theta'\in sp(\llbracket X\xleftarrow{\$}\{a_1: k_{1},...,a_n: k_{n} \}\rrbracket (\widehat{\theta}))$. We have $\theta'=\llbracket X\leftarrow k_i\rrbracket\theta$ for some $1\leq i\leq n$. 

      ($\Leftarrow$) By $\llbracket X\leftarrow k_i\rrbracket(\widehat{\theta})\models [\mathcal{B}_j^i]\psi$, $\theta'\models [\mathcal{B}_j^i]\psi$. So $\llbracket  X\xleftarrow{\$}\{a_1: k_{1},...,a_n: k_{n}\rrbracket (\widehat{\theta})\models [\mathcal{B}_j^i]\psi$.

      ($\Rightarrow$) By $\llbracket  X\xleftarrow{\$}\{a_1: k_{1},...,a_n: k_{n}\rrbracket (\widehat{\theta})\models [\mathcal{B}_j^i]\psi$, $\theta'\models [\mathcal{B}_j^i]\psi$, which implies that $\llbracket X\leftarrow k_i\rrbracket(\widehat{\theta})\models^I [\mathcal{B}_j^i]\psi$.

      So we proved the case.
\end{itemize}

\item (SEQ) $\theta \models \wpd( C_1 ;C_2, \phi)$ iff $\theta \models \wpd(C_1, \wpd(C_2, \phi ) )$   iff, by IH, $[\![ C_1]\!] (\theta ) \models \wpd(C_2, \phi )$ iff, by IH, $ [\![C_2 ]\!] ([\![ C_1]\!] (\theta ) \models ) \models \phi $ iff $[\![ C_1 ; C_2]\!] (\theta ) \models \phi $

\item (IF) $\theta \models \wpd(  if\ B\ then\ C_{1}\ else\ C_{2}, \phi  )$ iff $\theta \models (B \wedge \wpd(C_1, \phi))  \vee (\neg B \wedge \wpd(C_2, \phi)) $ iff $\theta \models B$ and $\theta \models \wpd(C_1, \phi)$ or $\theta \models \neg B$ and $\theta \models \wpd(C_2, \phi)$ iff $\theta \models B$ and $[\![C_1 ]\!] \theta \models \phi$ or $\theta \models \neg B$ and $[\![C_2 ]\!] \theta \models \phi$ iff $[\![ if\ B\ then\ C_{1}\ else\ C_{2} ]\!] \theta \models \phi $.

\item (WHILE) ($\Rightarrow$) Assume $\theta \models \wpd( while \mbox{ }B \mbox{ } do \mbox{ }C  , \phi)$, then from soundness and $\vdash \{\wpd(C,\phi)\} C\{ \phi\}$ we know $\models \{\wpd(C,\phi)\} C\{ \phi\}$. Therefore, $[\![   while \mbox{ }B \mbox{ } do \mbox{ }C ]\!] \theta \models \phi  $.

($ \Leftarrow $)    
Assume $[\![   while \mbox{ }B \mbox{ } do \mbox{ }C ]\!] \theta \models \phi  $.

 We prove by induction on the iteration times of the while loop.

\begin{itemize}

\item 
Suppose that the while loop is terminated after  $0$ time of execution.\footnote{The loop terminated after $k$ steps means that the longest branch of the execution of the while loop terminates after $k$ steps. } Then  $   \theta   \not\models B   $ and       $   [\![         while \mbox{ }B \mbox{ } do \mbox{ }C      ]\!]   ( \theta )  = \theta $. Therefore,    $\theta \models^I \phi$. Then we have $\theta   \models^I \neg B \wedge \phi$. Hence $\theta \models^I \psi_i$ for all $i$. Therefore, $\theta  \models^I \wpd( while \mbox{ }B \mbox{ } do \mbox{ }C  , \phi)$.

\item  
Suppose that the while loop is terminated after  $1$ time of execution. Then  $   \theta   \models B   $ , $ ([\![   C  ]\!] \circ \downarrow_{ B})(   \theta ) = [\![   C  ]\!] (   \theta )$ and $ [\![   C  ]\!] (   \theta ) \models \neg B \wedge \phi  $. 

Hence $[\![   C  ]\!] (   \theta  )  \models^I \psi_i$ for all $i$.

By I.H. we know $\theta \models \wpd( C, \psi_i)$. Therefore, $ \theta \models B\wedge \wpd( C, \psi_i)$, 

 $\theta \models^I \psi_i$ for all $i$. That is, $\theta \models^I \wpd( while \mbox{ }B \mbox{ } do \mbox{ }C  , \phi)$.

\item  By induction hypothesis, it holds that if iteration time is $k$, $\llbracket While\ B\ do\ C\rrbracket\models\phi$.

Suppose that the while loop is terminated after  $k+1$ times of execution. 

Let $\theta_1 \in sp (   ([\![   C   ]\!] \circ \downarrow_{ B})^k (\theta )  )$. 

\begin{itemize}

\item If $\theta_1 \models \neg B$, then $\theta_1$ is a state where the loop terminates after $k$ time of execution. Then by the induction hypothesis we know $\theta_1\models^I \psi_i$ for all $i$.

\item If $\theta_1 \models  B$, then $ ([\![   C  ]\!] \circ \downarrow_{ B})(  \theta_1 ) = [\![   C  ]\!] (  \theta_1 )$. Therefore, for all $\theta' \in sp ( [\![   C  ]\!] (\theta_1 ) )$,  $   \theta' \not\models B$, $\theta'  \models^I  \phi$ and $[\![   C  ]\!] ( \theta_1 )  \models^I \neg B\wedge \phi$. 
Hence $[\![   C  ]\!] (  \theta_1 )  \models^I \psi_i$ for all $i$. By I.H. we know $\theta_1 \models \wpd(C, \psi_i)$. From $\theta_1 \models B$ we deduce that

 $ \theta_1 \models^I \psi_i$ for all $i$.
\end{itemize}

This proves $ ([\![   C   ]\!] \circ \downarrow_{ B})^k ( \theta )   \models^I \psi_i $ for all $i$.  








\item Now we study the case in which  the loop never terminates. That is, there is an infinite sequence of pure cq-states $\theta_0, \theta_1, \theta_2,\ldots$  in which $\theta_0=\theta$ and $\theta_{j+1} \in sp ([\![ C ]\!](\theta_j) )$. We will show that $\theta_j \models \psi_i$ for all $i$ and $j$.

It's easy to see that $\theta_j \models B$ for all $j$ because otherwise the loop will terminate. Now we  prove that for all $j$, $\theta_j \models \psi_i$ for all $i$ by induction on $i$.

It's easy to see that $\theta_j \models \psi_0$ for all $j$. Assume for all $j$ it holds that $\theta_j \models \psi_k$. Then $\theta_{j-1} \models \wpd(C, \psi_k)$. Therefore, $\theta_{j-1} \models \psi_{k+1}$.
\end{itemize}

\item (UNITARY) $\theta \models  \wpd(U[\overline{q}], \phi ) $ iff $\theta \models [U_{\overline{q}}] \phi$ iff $  U[\overline{q}] \theta \models \phi $.

\item (MEASURE) $\theta \models \wpd(X \leftleftarrows q_j,\phi)  $ iff $\theta \models ( P^1_j  \vee  [\mathfrak{P}^0_j] \phi[X/0]   ) \wedge ( P^0_j  \vee  [\mathfrak{P}^1_j] \phi[X/1]  ) $.    Now we prove it by case analysis. 

\begin{enumerate}
\item If $\theta \models P^1_j$, then $\theta \models ( P^1_j  \vee  [\mathfrak{P}^0_j] \phi[X/0]   ) \wedge ( \mathfrak{P}^0_j  \vee  [P^1_j] \phi[X/1]  ) $ iff $[\![X \leftarrow M[q_j]]\!] \theta  \models \phi $

\item  If $\theta \models P^0_j$, then $\theta \models ( P^1_j  \vee  [\mathfrak{P}^0_j] \phi[X/0]   ) \wedge ( P^0_j  \vee  [\mathfrak{P}^1_j] \phi[X/1]  ) $ iff  $[\![X \leftarrow M[q_j]]\!] \theta \models \phi $

\item  If $\theta \models \neg  P^0_j  \wedge \neg P^1_j $, then $\theta \models ( P^1_j  \vee  [\mathfrak{P}^0_j] \phi[X/0]   ) \wedge ( P^0_j  \vee  [\mathfrak{P}^1_j] \phi[X/1]  ) $ iff  $[\![X \leftleftarrows q_j ]\!] \theta \models \phi $

\end{enumerate}

\end{itemize}

\end{proof}

\begin{proposition*}\label{proof:prop-wp}
It holds that $ \vdash \{   \wpd(C, \phi) \}   C\{\phi\}$.
\end{proposition*}
\begin{proof} 
We use structural induction on $C$. We only show these non-trivial cases below. The cases of (SKIP),(AS), (PAS), (UNITARY) and  (MEASURE)  are trivial.

\begin{itemize}

\item (SEQ)  We have $\vdash \{ \wpd(C_1, \wpd(C_2, \phi))  \} C_1 \{   \wpd(C_2, \phi)  \} $ and $ \vdash  \{  \wpd(C_2, \phi) \} C_2 \{\phi \} $ by inductive hypothesis. Using SEQ, we have $\vdash \{ \wpd(C_1, \wpd(C_2, \phi))  \} C_1 ;C_2 \{    \phi   \} $.

\item (IF) By inductive hypothesis we have $ \{  \wpd(C_1, \phi) \}C_1 \{\phi \} $ and $ \{  \wpd(C_2, \phi) \}C_2\{\phi \} $. Then by (CONS) we know $ \{ B \wedge ( \wpd(C_1, \phi) \wedge B ) \vee (\wpd(C_2, \phi) \wedge  \neg B )) \}C_1 \{\phi \} $ and $ \{ \neg  B \wedge ( \wpd(C_1, \phi) \wedge B ) \vee (\wpd(C_2, \phi) \wedge  \neg B ))) \}C_2\{\phi \} $. Now, by (IF) we have $\vdash  \{( \wpd(C_1, \phi) \wedge B ) \vee (\wpd(C_2, \phi) \wedge  \neg B )  \}    if\ B\ then\ C_{1}\ else\ C_{2} \{\phi \} $.

\item (WHILE) It is easy to see that $\models ( B\wedge  \bigwedge\limits_{k \geq 0}  \psi_k  ) \rightarrow (B \wedge  \bigwedge\limits_{k \geq 1}  \wpd(C, \psi_i)  )   $. By inductive hypothesis, $\vdash  \{\bigwedge\limits_{k \geq 1}  \wpd(C, \psi_i)\} C \{  \bigwedge\limits_{k \geq 1}    \psi_i \}$. Now by (CONS) we have $\vdash  \{B \wedge  \bigwedge\limits_{k \geq 1}  \wpd(C, \psi_i)\} C \{  \bigwedge\limits_{k \geq 1}    \psi_i \}$. Note that $\models \bigwedge\limits_{k \geq 1}    \psi_i  \leftrightarrow \bigwedge\limits_{k \geq 0}    \psi_i $. Then by (CONS) we have $\vdash  \{B \wedge  \bigwedge\limits_{k \geq 1}  \wpd(C, \psi_i)\} C \{  \bigwedge\limits_{k \geq 0}    \psi_i \}$.\\
 That is, $\vdash  \{B \wedge     \wpd( while \mbox{ }B \mbox{ } do \mbox{ }C  , \phi)  \} C \{  \wpd( while \mbox{ }B \mbox{ } do \mbox{ }C  , \phi)   \}$.\\
 
 Then by (WHILE) we know   \\
 
 $\vdash  \{   \wpd( while \mbox{ }B \mbox{ } do \mbox{ }C  , \phi)  \} while \mbox{ }B \mbox{ } do \mbox{ }C \{  \wpd( while \mbox{ }B \mbox{ } do \mbox{ }C  , \phi)  \wedge \neg B  \}$.\\
 
 Then by the definition of $\wpd( while \mbox{ }B \mbox{ } do \mbox{ }C  , \phi)$ we know \\
  $\vdash  \{   \wpd( while \mbox{ }B \mbox{ } do \mbox{ }C  , \phi)  \} while \mbox{ }B \mbox{ } do \mbox{ }C \{  \phi  \}$. 
\end{itemize}
\end{proof}

\begin{lemma*}\label{proof:lemma-wpd}
For all $i\geq 0$ and all pure cq-state $\theta$, if $\theta \models^I \neg \wpd(C^0 , \neg B ) \wedge \ldots \wedge  \neg \wpd(C^{i-1} , \neg B) \wedge   \wpd(C^{i} , \neg B ) $, then 
\[
[\![  \texttt{while} \mbox{ } B \mbox{ } \texttt{do} \mbox{ } C  ]\!] (\widehat{\theta}) =  [\![  ( \texttt{if}\mbox{ } B \mbox{ } \texttt{then}\mbox{ }  C \mbox{ }  \texttt{else}\mbox{ }  \texttt{skip} )^{i} ]\!] (\widehat{\theta}).
\]
\end{lemma*}

\begin{proof} 
We first prove the case where $i=0$. Then we have $\theta \models^I \wpd(C^{0} , \neg B)$. That is, $\theta \models^I \neg B$. It is trivial to show that $[\![  \texttt{while} \mbox{ } B \mbox{ } \texttt{do} \mbox{ } C  ]\!] (\widehat{\theta}) =  [\![  ( \texttt{if}\mbox{ } B \mbox{ } \texttt{then}\mbox{ }  C \mbox{ }  \texttt{else}\mbox{ }  \texttt{skip} )^{0} ]\!] (\widehat{\theta})$. 

We then prove the case where $i=1$. Then we have $\theta \models^I \neg \wpd(C^{0}, \neg B)\wedge  \wpd(C,\neg B)$, which implies that $\theta\models^I\neg B$ and for all $\theta'\in sp(\llbracket C\rrbracket\widehat{\theta})$, $\theta'\models^I \neg B$. Then it is easy to see that $[\![  \texttt{while} \mbox{ } B \mbox{ } \texttt{do} \mbox{ } C  ]\!] (\widehat{\theta}) =  [\![  ( \texttt{if}\mbox{ } B \mbox{ } \texttt{then}\mbox{ }  C \mbox{ }  \texttt{else}\mbox{ }  \texttt{skip} )^{1} ]\!] (\widehat{\theta})$.

We now prove the case where $i=2$, other cases are similar.
Assume $i=2$. Then we have $\theta\models \neg \wpd(C^0,\neg B) \wedge \neg \wpd(C^1, \neg B) \wedge \wpd(C^2 ,\neg B)$. 	Then we know $\theta \not\models^I \wpd( \texttt{skip} , \neg B)$, $\theta \not\models^I \wpd(  C, \neg B)$  and $\theta\models^I \wpd(C;C ,\neg B)$. Therefore, $\theta\models^I \neg B$, $[\![C ]\!]  (\widehat{\theta}) \not \models^I  \neg B$ and $[\![C;C ]\!]  (\widehat{\theta}) \models^I  \neg B$. We then know $\downarrow_{  B} ( [\![C;C ]\!]  (\widehat{\theta}))  =  \textbf{0}$.

It's easy to see that $sp(     [\![   C   ]\!] \circ \downarrow_{ B}  ( \widehat{\theta}) ) \subseteq sp(    [\![   C   ]\!]    ( \widehat{\theta})  ) $. Then from $\theta \models^I \wpd(C^2 , \neg B)$ we deduce  $[\![   C^2   ]\!] \theta \models^I \neg B $, which implies $[\![   (C \circ  \downarrow_{ B})^2   ]\!] \theta \models^I \neg B $.

Then $\displaystyle\sum_{i = 0}^{\infty} \downarrow_{\neg B}  (  (  [\![   C   ]\!] \circ \downarrow_{ B})^i ( \widehat{\theta}) )  =   \downarrow_{\neg B}  (  (  [\![   C   ]\!] \circ \downarrow_{ B})^0 ( \widehat{\theta}))  + \downarrow_{\neg B}  (  (  [\![   C   ]\!] \circ \downarrow_{ B})^1 ( \widehat{\theta}) )  + \downarrow_{\neg B}  (  (  [\![   C   ]\!] \circ \downarrow_{ B})^2 ( \widehat{\theta} ) ) + \downarrow_{\neg B}  (  (  [\![   C   ]\!] \circ \downarrow_{ B})^3 ( \widehat{\theta}) ) + \ldots  =  \downarrow_{\neg B}  (  (  [\![   C   ]\!] \circ \downarrow_{ B})^0 ( \widehat{\theta}) )  + \downarrow_{\neg B}  (  (  [\![   C   ]\!] \circ \downarrow_{ B})^1 ( \widehat{\theta})  + \downarrow_{\neg B}  (  (  [\![   C   ]\!] \circ \downarrow_{ B})^2 ( \widehat{\theta}) )  =  \downarrow_{ \neg B} ( \widehat{\theta} ) +  \downarrow_{ \neg B}\circ  [\![   C   ]\!] \circ \downarrow_{ B} (  \widehat{\theta}  )  + ([\![   C   ]\!] \circ \downarrow_{ B} )^2  (\widehat{\theta})  $.

On the other hand,  $         [\![\texttt{if}  \mbox{ }  B   \mbox{ }  \texttt{then}   \mbox{ }  C  \mbox{ }  \texttt{else}  \mbox{ } \texttt{skip}  ]\!]^{2}   (\widehat{\theta})  =  $\\
$ [\![\texttt{if}  \mbox{ }  B   \mbox{ }  \texttt{then}   \mbox{ }  C  \mbox{ }  \texttt{else}  \mbox{ } \texttt{skip}  ]\!]   (   [\![   C   ]\!] \circ \downarrow_{ B} ( \widehat{\theta} )  + \downarrow_{ \neg B} ( \widehat{\theta}) )  =  $\\
$[\![   C   ]\!] \circ \downarrow_{ B} ([\![   C   ]\!] \circ \downarrow_{ B} ( \widehat{\theta})  + \downarrow_{ \neg B} ( \widehat{\theta}) )  + \downarrow_{ \neg B} ([\![   C   ]\!] \circ \downarrow_{ B} ( \widehat{\theta})  + \downarrow_{ \neg B} ( \widehat{\theta}) ) = $\\
$[\![   C   ]\!] \circ \downarrow_{ B}  [\![   C   ]\!] \circ \downarrow_{ B} ( \widehat{\theta} )  +  [\![   C   ]\!] \circ \downarrow_{ B} (\downarrow_{ \neg B} ( \widehat{\theta} ) )  + \downarrow_{ \neg B} [\![   C   ]\!] \circ \downarrow_{ B} ( \widehat{\theta} )  + \downarrow_{ \neg B} (\downarrow_{ \neg B} ( \widehat{\theta} ) )  = $\\
$ ([\![   C   ]\!] \circ \downarrow_{ B} )^2  ( (\widehat{\theta}) )  +  \textbf{0}  + \downarrow_{ \neg B} [\![   C   ]\!] \circ \downarrow_{ B} ( \widehat{\theta} )  +   \downarrow_{ \neg B} ( \widehat{\theta})  =$\\
$   \downarrow_{ \neg B} ( \widehat{\theta})  +  \downarrow_{ \neg B} [\![   C   ]\!] \circ \downarrow_{ B} ( \widehat{\theta})  + ([\![   C   ]\!] \circ \downarrow_{ B} )^2  ( \widehat{\theta})  $.

\end{proof}

\section{Proofs in Section \ref{sec:probabilistic}}
\begin{lemma*} \label{proof:lemma-pt}
Given an arbitrary probabilistic cq-state $\Theta $, an interpretation $I$, a command $C$ and a real expression $r$,   

   $$ [\![  pt  ( C,  r ) ]\!]^I_{\Theta } =  [\![ r ]\!]^I_{ [\![ C]\!] \Theta} $$
\end{lemma*}

\begin{proof} 

\begin{enumerate}

\item $ [\![ \pt(C,  a  ) ]\!]^I_{\Theta } = [\![  a ]\!]^I_{\Theta } =  a =  [\![  a ]\!]^I_{ [\![ C]\!] \Theta}$.

\item $ [\![ \pt(C,  \mathfrak{x}   ) ]\!]^I_{\Theta} = [\![  \mathfrak{x}   ]\!]^I_{\Theta } =  I( \mathfrak{x}  ) =  [\![  \mathfrak{x}   ]\!]^I_{ [\![ C]\!] \Theta }$.

\item $ [\![ \pt(C, r_1\mbox{ } aop \mbox{ } r_2 ) ]\!]^I_{\Theta} =    [\![   \pt(C, r_1) \mbox{ } aop \mbox{ } \pt(C, r_2)     ]\!]^I_{\Theta} =   $ \\  $  [\![   \pt(C, r_1) \mbox{ }   ]\!]^I_{\Theta}\ aop\ [\![    \mbox{ } \pt(C, r_2)     ]\!]^I_{\Theta}  =        [\![  r_1  ]\!]^I_{ [\![ C]\!] \Theta}\   aop\    [\![  r_2  ]\!]^I_{ [\![ C]\!] \Theta}  $.    

\item Similarly, $\llbracket \pt(C,\sum\limits_{i\in K}r_i)\rrbracket_\Theta^I=\llbracket\sum\limits_{i\in K}r_i\rrbracket^I_{\llbracket C\rrbracket\Theta}$.

\item    $   [\![  pt  ( \texttt{skip},  \mathbb{P} (\phi) ) ]\!]^I_{\Theta}  =   [\![ \mathbb{P} (  \phi  )]\!]^I_{\Theta } =    [\![ \mathbb{P} (  \phi  )]\!]^I_{ [\![ \texttt{skip}]\!] \Theta} $. 

\item   $ [\![  pt  ( X\leftarrow E ,  \mathbb{P} (\phi) ) ]\!]^I_{\Theta } = [\![ \mathbb{P} (   \phi[X/E] )]\!]^I_{\Theta }    =  [\![ \mathbb{P} (\wpd(    X\leftarrow E    , \phi) )]\!]^I_{\Theta } $. Since for each $\theta\in\mathbb{PS}$, $\theta \models \phi[X/E] $ iff $ \theta\models  \wpd(X\leftarrow E, \phi) $ iff $  [\![  X\leftarrow E ]\!] \widehat{\theta} \models  \phi  $, then it can inferred that $ [\![ \mathbb{P} (   \phi[X/E] )]\!]^I_{\Theta } =  [\![ \mathbb{P} (   \phi  )]\!]^I_{  [\![  X\leftarrow E ]\!] \Theta }$.

\item PAS: For the sake of simplicity, we assume $n=2$. No generality is lost with this assumption. We have $ [\![  pt  (  X\xleftarrow[]{\$}  \{a_1: k_1,  a_2:k_2\} ,  \mathbb{P} (\phi) ) ]\!]^I_{\Theta}  =   [\![  a_1  \mathbb{P} (\phi[X/k_1 ] \wedge \neg \phi[X/k_2 ] ) +  a_2 \mathbb{P} ( \neg \phi[X/k_1 ] \wedge   \phi[X/k_2 ] ) + (a_1 +a_2 )  \mathbb{P} (\phi[X/k_1 ] \wedge  \phi[X/k_2 ] ) ]\!]^I_{\Theta} $. Without loss generality, assume $sp(\Theta) =\{\theta_1, \theta_2, \theta_3, \theta_4\}$, $\theta_1 \models^I \phi[X/k_1 ] \wedge \neg \phi[X/k_2 ] $, $\theta_2 \models^I \neg \phi[X/k_1 ] \wedge   \phi[X/k_2 ] $, $\theta_3 \models^I \phi[X/k_1 ] \wedge   \phi[X/k_2 ] $, $\theta_4 \models^I \neg \phi[X/k_1 ] \wedge \neg \phi[X/k_2 ] $, $\Theta(S_i) = b_i$. Then we have $  \![  a_1  \mathbb{P} (\phi[X/k_1 ] \wedge \neg \phi[X/k_2 ] ) +  a_2 \mathbb{P} ( \neg \phi[X/k_1 ] \wedge   \phi[X/k_2 ] ) + (a_1 +a_2 )  \mathbb{P} (\phi[X/k_1 ] \wedge  \phi[X/k_2 ] ) ]\!]^I_{\Theta } =  a_1 b_1 + a_2 b_2 + (a_1 +a_2)b_3$.

We further have  $[\![X\xleftarrow[]{\$}  \{a_1: k_1,  a_2:k_2\}]\!] (\theta_1) \models^I \mathbb{P} (\phi) = a_1 $,  $[\![X\xleftarrow[]{\$}  \{a_1: k_1,  a_2:k_2\}]\!] (\theta_2) \models^I \mathbb{P} (\phi) = a_2 $,  $[\![X\xleftarrow[]{\$}  \{a_1: k_1,  a_2:k_2\}]\!] (\theta_3) \models^I \mathbb{P} (\phi) = a_1 +a_2 $ and $[\![X\xleftarrow[]{\$}  \{a_1: k_1,  a_2:k_2\}]\!] (\theta_4) \models^I \mathbb{P} (\phi) = 0 $. Then we know $[\![   X\xleftarrow[]{\$}  \{a_1: k_1,  a_2:k_2\}    ]\!] (\Theta) \models^I \mathbb{P} (\phi) = a_1 b_1 + a_2 b_2 + (a_1 +a_2)b_3  $. This means that $ [\![ \mathbb{P} (  \phi  )]\!]^I_{[\![    X\xleftarrow[]{\$}  \{a_1: k_1,  a_2:k_2\}         ]\!] \Theta} = a_1 b_1 + a_2 b_2 + (a_1 +a_2)b_3 =  [\![  pt  (  X\xleftarrow[]{\$}  \{a_1: k_1,  a_2:k_2\} ,  \mathbb{P} (\phi) ) ]\!]^I_{\Theta} $.

\item UNITARY: Trivial.
  
\item MEASURE:  $[\![ X \leftleftarrows q_j  ]\!] \widehat{\theta}= p_0 \widehat{\theta_0} + p_1 \widehat{\theta_1}$, where 

$p_0 = \sqrt{ \Tr(  P^0_j \theta_q \theta_q^{\dagger} ) }$, $p_1 = \sqrt{ \Tr(  P^1_j \theta_q \theta_q^{\dagger} ) }$,  $\theta_{0,c} =  \theta_c[X\mapsto 0] $,  $\theta_{1,c} = \theta_c[X\mapsto 1] $, $\theta_{0,q} =  \frac{P^0_j \theta_q}{\sqrt{p_0}}$, $\theta_{1,q} =  \frac{P^1_j \theta_q}{\sqrt{p_1}}$.\\

If $p_0=0$, then $p_1=1$ and $[\![ \mathbb{P} (\phi) ]\!]_{ [\![ X \leftleftarrows q_j  ]\!] \widehat{\theta} } = [\![ \mathbb{P} (\phi) ]\!]_{ \theta_1}$. Note that in this case $[\![ \mathbb{P} (\phi) ]\!]_{ \theta_1} =[\![  \mathbb{P}(    [\mathfrak{P}^1_j ]   \phi [X/1]            )   ]\!]_{\theta}   $. We then know   $  [\![ \mathbb{P} (\phi) ]\!]_{ [\![ X \leftleftarrows q_j  ]\!] \widehat{\theta}}  = [\![ (\top \Rightarrow  P^0_j )    \mathbb{P}(    [\mathfrak{P}^0_j ]\phi [X/0]       )  +   (\top \Rightarrow  P^1_j )      \mathbb{P}(    [\mathfrak{P}^1_j ]   \phi [X/1]            )   ]\!]_{\theta}  $.

If $p_1=0$, then $p_0=1$ and $[\![ \mathbb{P} (\phi) ]\!]_{ [\![ X \leftleftarrows q_j  ]\!] \widehat{\theta} } = [\![ \mathbb{P} (\phi) ]\!]_{ \theta_0}$. Note that in this case $[\![ \mathbb{P} (\phi) ]\!]_{ \theta_0} =[\![  \mathbb{P}(    [\mathfrak{P}^0_j ]   \phi [X/0]            )   ]\!]_{\theta}   $. We then know   $  [\![ \mathbb{P} (\phi) ]\!]_{ [\![ X \leftleftarrows q_j  ]\!] \widehat{\theta} }  = [\![ (\top \Rightarrow  P^0_j )    \mathbb{P}(    [\mathfrak{P}^0_j ]\phi [X/0]       )  +   (\top \Rightarrow  P^1_j )      \mathbb{P}(    [\mathfrak{P}^1_j ]   \phi [X/1]            )   ]\!]_{\theta}     $.

If $p_0, p_1 >0$, then $[\![  \mathbb{P}(    [\mathfrak{P}^0_j ]   \phi [X/0]            )   ]\!]_{\theta} =   [\![  \mathbb{P}(       \phi [X/0]            )   ]\!]_{    \frac{( P^0_j \otimes I_{-j} ) \theta  }{\sqrt{p_0}}  } =  [\![  \mathbb{P}(       \phi [X/0]            )   ]\!]_{    ( \theta_{c}, \theta_{0,q} )  } =   [\![  \mathbb{P}(       \phi              )   ]\!]_{ \theta_0} $.

Similarly, $[\![  \mathbb{P}(    [\mathfrak{P}^1_j ]   \phi [X/1]            )   ]\!]_{\theta} =  [\![  \mathbb{P}(       \phi              )   ]\!]_{ \theta_1}$.\\

We then know  $  [\![ \mathbb{P} (\phi) ]\!]_{ [\![ X \leftleftarrows q_j  ]\!] \widehat{\theta} }  = [\![ (\top \Rightarrow  P^0_j )    \mathbb{P}(    [\mathfrak{P}^0_j ]\phi [X/0]       )  +   (\top \Rightarrow  P^1_j )      \mathbb{P}(    [\mathfrak{P}^1_j ]   \phi [X/1]            )   ]\!]_{\theta}     $.

 \item SEQ: $   [\![  pt  ( C_1; C_2,  \mathbb{P} (\phi) ) ]\!]^I_{\Theta } =   [\![   \pt(C_1, \pt( C_2, \mathbb{P} (\phi)) )]\!]^I_{\Theta } = $\\  $   [\![    \pt(C_2, \mathbb{P} (\phi)    )  ]\!]^I_{  [\![  C_1 ]\!]  \Theta }  =   [\![ \mathbb{P} (\phi)]\!]^I_{ [\![  C_2 ]\!]  [\![  C_1 ]\!]  \Theta }  =    [\![ \mathbb{P} (\phi)  ]\!]^I_{    [\![  C_1 ;C_2 ]\!]  \Theta }    $.
 
 \item IF: $[\![  pt  (\texttt{if}\ B\ \texttt{then}\ C_{1}\ \texttt{else}\ C_{2}  ,  \mathbb{P} (\phi) ) ]\!]^I_{\Theta } =   $\\ 
 
  $ [\![ \pt(C_1,   \mathbb{P} (\phi    ) )/B   +   \pt(C_2,  \mathbb{P} (\phi  )  )/(\neg B)  ]\!]^I_{\Theta } =        $ \\
  
  $ [\![ \pt(C_1,  \mathbb{P} (\phi  )  )/B      ]\!]^I_{\Theta }  +      [\![   \pt(C_2,  \mathbb{P} (\phi  )  )/(\neg B)  ]\!]^I_{\Theta }   = $ \\
  
  $ [\![ \pt(C_1,  \mathbb{P} (\phi  )  )      ]\!]^I_{ \downarrow_{B} \Theta }  +      [\![   \pt(C_2,  \mathbb{P} (\phi  )  )   ]\!]^I_{ \downarrow_{\neg B}\Theta}   = $ \\
    
  $ [\![    \mathbb{P} (\phi  )       ]\!]^I_{[\![ C_1 ]\!] \downarrow_{B} \Theta}  +      [\![      \mathbb{P} (\phi  )    ]\!]^I_{[\![ C_2 ]\!] \downarrow_{\neg B} \Theta }   = $ \\
  
  $ [\![    \mathbb{P} (\phi  )       ]\!]^I_{   [\![ C_1 ]\!] \downarrow_{B} \Theta +      [\![ C_2 ]\!] \downarrow_{\neg B} \Theta  }    =   $   
  $ [\![    \mathbb{P} (\phi  )       ]\!]^I_{   [\![   \texttt{if}\ B\ \texttt{then}\ C_{1}\ \texttt{else}\ C_{2}     ]\!]      \Theta }       $

\item WHILE: Without loss of generality, let $sp(\Theta) = \{\theta_{0,\infty}, \theta_{0,0},\theta_{0,1},\ldots  \}$, in which  $\theta_{0,i} \models \wpd(i) $ and $\theta_{0,\infty} \models   \wpd(\infty )$. Equivalently, we may let $\Theta_{\theta_{0,i}}= \downarrow_{\wpd(i)} (\Theta) $ and $\Theta_{\theta_{0,\infty}}= \downarrow_{\wpd(\infty)} (\Theta)$.

Then  we know  
$\Theta(\theta_{0,i}) =  [\![  \mathbb{P}(   \wpd(i)  )  ]\!]_{\Theta}$, $\Theta(\theta_{0,\infty}) =  [\![  \mathbb{P}(    \wpd(\infty)  )  ]\!]_{\Theta}$.

  That is, $\Theta=  [\![  \mathbb{P}(   \wpd(\infty)  )  ]\!]_{\Theta}  \Theta_{\theta_{0,\infty}}  + \displaystyle\sum_{i = 0}^{\infty}  [\![  \mathbb{P}(  \wpd(i)  )  ]\!]_{\Theta } \Theta_{\theta_{0,i}}$.

Then $  [\![   \mathbb{P}(\phi  ) ]\!]_{ [\![ WL  ]\!] \Theta }   $

$$   =[\![   \mathbb{P}(\phi  ) ]\!]_{ [\![ WL  ]\!] (\Theta(\theta_{0,\infty}) \Theta_{\theta_{0,\infty}}  +\displaystyle\sum_{i = 0}^{\infty} \Theta(\theta_{0,i}) \Theta_{\theta_{0,i}})}  $$

$$ =[\![   \mathbb{P}(\phi  ) ]\!]_{ [\![ WL ]\!] (  \Theta(S_{0,\infty}) \Theta_{\theta_{0,\infty}})} +  [\![   \mathbb{P}(\phi  ) ]\!]_{ [\![WL  ]\!]  \displaystyle\sum_{i = 0}^{\infty} \Theta(\theta_{0,i}) \Theta_{\theta_{0,i}}} $$

$$=  \Theta(\theta_{0,\infty})  [\![   \mathbb{P}(\phi  ) ]\!]_{ [\![ WL ]\!] (   \Theta_{\theta_{0,\infty }}    ) } +  \displaystyle\sum_{i = 0}^{\infty} \Theta(\theta_{0,i}) [\![   \mathbb{P}(\phi  ) ]\!]_{  [\![ WL  ]\!] \Theta_{\theta_{0,i}}  }  $$

By Lemma~\ref{key lemma}, we know $[\![ WL  ]\!] \Theta_{\theta_{0,i}}  =   [\![ (IF)^i   ]\!] \Theta_{\theta_{0,i}}$.  \\

Therefore, 

$ [\![ \mathbb{P}(\phi  ) ]\!]_{  [\![WL  ]\!] \Theta_{\theta_{0,i}}} =   [\![ \mathbb{P}(\phi  ) ]\!]_{  [\![ (IF)^i   ]\!] \Theta_{\theta_{0,i}}}$

By induction hypothesis we know
$$  [\![ \mathbb{P}(\phi  ) ]\!]_{  [\![ (  IF)^i   ]\!] \Theta_{\theta_{0,i}}}  =  [\![\pt( (  IF)^i  ,\mathbb{P}(\phi  )  ) ]\!]_{\Theta_{\theta_{0,i}}}  $$

Then we have 
$$  [\![ \mathbb{P}(\phi  ) ]\!]_{ [\![WL  ]\!] \Theta_{\theta_{0,i}}    }  =  [\![\pt( (IF)^i  ,\mathbb{P}(\phi))]\!]_{    \Theta_{\theta_{0,i}}}$$

Moreover, $$  [\![\pt( (IF)^i  ,\mathbb{P}(\phi  )  ) ]\!]_{    \Theta_{\theta_{0,i}}    } =[\![\pt( (IF)^i  ,\mathbb{P}(\phi  )  ) ]\!]_{    \downarrow_{\wpd(i)}(\Theta)} $$

$$ =  [\![\pt( (IF)^i  ,\mathbb{P}(\phi  )  ) / (   \wpd(i)     ) ]\!]_{    \Theta     }  $$

At this stage we know $  [\![   \mathbb{P}(\phi  ) ]\!]_{ [\![WL ]\!] \Theta }   = $

$$  \Theta(\theta_{0,\infty})  [\![   \mathbb{P}(\phi  ) ]\!]_{ [\![WL  ]\!] (   \Theta_{\theta_{0,\infty }}    ) } +  \displaystyle\sum_{i = 0}^{\infty} \Theta(\theta_{0,i}) [\![   \mathbb{P}(\phi  ) ]\!]_{  [\![ WL ]\!] \Theta_{\theta_{0,i}}}  $$

in which $ \displaystyle\sum_{i = 0}^{\infty} \Theta(\theta_{0,i}) [\![   \mathbb{P}(\phi) ]\!]_{  [\![ WL  ]\!] \Theta_{\theta_{0,i}}} = $

$ [\![  \displaystyle\sum_{i = 0}^{\infty} (  \mathbb{P}(    \wpd(i)    )      ( \pt( (  IF)^i  ,\mathbb{P}(\phi  )  ) / (    \wpd(i)   )) )  ]\!]_{\Theta } $

Note that $SUM$ is short for $$\displaystyle\sum_{i = 0}^{\infty} (  \mathbb{P}(    \wpd(i)    )      ( \pt( (  IF)^i  ,\mathbb{P}(\phi  )  ) / (    \wpd(i)   )) ).$$ 

Then $  [\![   \mathbb{P}(\phi  ) ]\!]_{ [\![WL ]\!] \Theta }   = $
$$  \Theta(\theta_{0,\infty})  [\![   \mathbb{P}(\phi  ) ]\!]_{ [\![WL  ]\!] (   \Theta_{\theta_{0,\infty}})} + [\![ SUM ]\!]_{\Theta} $$
$$=  [\![    \mathbb{P}(\wpd(\infty)  )    ]\!]_{\Theta} [\![   \mathbb{P}(\phi  ) ]\!]_{ [\![WL  ]\!] (   \Theta_{\theta_{0,\infty }}    ) } + [\![ SUM ]\!]_{\Theta}  .$$

\ \\

It remains to study $[\![ \mathbb{P}(\phi  ) ]\!]_{  [\![ WL ]\!] \Theta_{\theta_{0,\infty}}  }$.

Note that $ [\![ WL ]\!]  = [\![ IF;WL ]\!]$.
 
 We then know $[\![ \mathbb{P}(\phi  ) ]\!]_{  [\![WL ]\!] \Theta_{\theta_{0,\infty}}  } $\\
 $= [\![ \mathbb{P}(\phi  ) ]\!]_{  [\![ IF; WL  ]\!] \Theta_{\theta_{0,\infty}}  }$\\
  $= [\![ \mathbb{P}(\phi  ) ]\!]_{  [\![  WL  ]\!] [\![ IF ]\!]    \Theta_{\theta_{0,\infty}}  }  $\\
 $= [\![ \mathbb{P}(\phi  ) ]\!]_{  [\![  WL  ]\!] [\![ \texttt{if}\ B\ \texttt{then}\ C \ \texttt{else} \mbox{ }  \texttt{skip}   ]\!]    \Theta_{\theta_{0,\infty}}  }  $\\
  $= [\![ \mathbb{P}(\phi) ]\!]_{  [\![  WL  ]\!] [\![   C     ]\!]    \Theta_{\theta_{0,\infty}}}  $.

 Here we also note that $[\![   C     ]\!]    \Theta_{\theta_{0,\infty}}= [\![   C     ]\!]\downarrow_{\wpd(\infty)}   ( \Theta )$.

Without loss of generality,  let $sp([\![   C     ]\!]  \Theta_{\theta_{0,\infty}} ) = \{\theta_{1,\infty}, \theta_{1,0},\theta_{1,1},\ldots,  \} $, in which  $\theta_{1,i} \models \wpd(i)  $ and $\theta_{1,\infty} \models  \wpd(\infty)$.

Then by repeating the reasoning on the cases of  $\theta_{0,i}$, we know \\

$  [\![   \mathbb{P}(\phi)]\!]_{ [\![ WL  ]\!]   [\![   C     ]\!]    \Theta_{\theta_{0,\infty}}} = $

$$  [\![    \mathbb{P}(\wpd(\infty)  )    ]\!]_{[\![   C     ]\!]    \Theta_{\theta_{0,\infty}} } [\![   \mathbb{P}(\phi  ) ]\!]_{ [\![WL  ]\!] (   \Theta_{\theta_{1,\infty }}    ) } + [\![ SUM ]\!]_{[\![   C     ]\!]    \Theta_{\theta_{0,\infty}} }  =$$

$$  [\![    \mathbb{P}(\wpd(\infty)  )    ]\!]_{[\![   C     ]\!]  \downarrow_{\wpd(\infty)}  (\Theta)  } [\![  \mathbb{P}(\phi  ) ]\!]_{ [\![WL  ]\!] (   \Theta_{\theta_{1,\infty }}    ) } + [\![ SUM ]\!]_{[\![   C     ]\!]    \downarrow_{\wpd(\infty)} (\Theta)}  $$

By induction hypothesis we know 
$$  [\![ \mathbb{P}(\wpd(\infty)  )    ]\!]_{[\![   C     ]\!]  \downarrow_{\wpd(\infty)}  (\Theta)  }  = [\![\pt(C, \mathbb{P}(\wpd(\infty)  ))    ]\!]_{ \downarrow_{\wpd(\infty)}  (\Theta)}$$
$$=  [\![\pt(C, \mathbb{P}(\wpd(\infty)  ))  /\wpd(\infty)  ]\!]_{ \Theta } $$

and $[\![ SUM ]\!]_{[\![   C     ]\!]    \downarrow_{\wpd(\infty)} (\Theta) } =  [\![\pt(C,  SUM)  /\wpd(\infty)  ]\!]_{ \Theta } $.

\ \\

It then remains to study $[\![ \mathbb{P}(\phi  ) ]\!]_{  [\![ WL ]\!] \Theta_{\theta_{1,\infty}}}$.

By repeating the reasoning on the cases of  $\theta_{0,i}$, we know \\

$[\![   \mathbb{P}(\phi  ) ]\!]_{ [\![ WL  ]\!]   [\![   C     ]\!]    \Theta_{\theta_{1,\infty}}   }   = $
  
  $$  [\![    \mathbb{P}(\wpd(\infty)  )    ]\!]_{[\![   C     ]\!]    \Theta_{\theta_{1,\infty}} } [\![   \mathbb{P}(\phi  ) ]\!]_{ [\![WL  ]\!] (   \Theta_{\theta_{2,\infty }}    ) } + [\![ SUM ]\!]_{[\![   C     ]\!]    \Theta_{\theta_{1,\infty}} }  $$
  
  in which  $ \Theta_{\theta_{2,\infty }} = \downarrow_{\wpd(\infty)} ([\![   C     ]\!] ( \Theta_{\theta_{1,\infty }} ) ) $.

\ \\

Let $ \Theta_{\theta_{i+1,\infty }} = \downarrow_{\wpd(\infty)} ([\![   C     ]\!] ( \Theta_{\theta_{i,\infty }} ) ) $.

Repeat the above procedure to infinity, we get $ [\![  \pt(WL, \mathbb{P}(\phi  ) )  ]\!]_{\Theta }=[\![   \mathbb{P}(\phi   )  ]\!]_{ [\![   WL  ]\!] (\Theta) } $.

\item (SKIP-cq) Trivial.

\item (AS-cq) $[\![ \phi [X /E] \Rightarrow Q ]\!]^I_{\Theta} = \sum\limits_{\theta\models \phi[X/E]} \Theta(\theta) [\![Q ]\!]^I_{\theta} $.\\

By what we proved in Theorem~\ref{th.precondition}, it holds that $ \theta\models^I \phi[X/E]$ iff $ [\![X\leftarrow E]\!] \theta \models^I \phi$ for each $\theta\in\mathbb{PS}$.\\

Let $\Theta'= [\![X\leftarrow E ]\!] \Theta$. Then $\Theta'(\theta') =\sum \{\Theta(\theta): [\![X\leftarrow E ]\!] \widehat{\theta} =\widehat{\theta'}\}$.\\

Note that $[\![Q]\!]_{\theta} = [\![Q]\!]_{\theta'} $ since $\theta_q= \theta'_q$.

Therefore, $  [\![ \phi \Rightarrow Q ]\!]^I_{[\![X\leftarrow E ]\!] \Theta  } =  \sum\limits_{\theta'\models \phi}  \Theta'(\theta')[\![Q]\!]^I_{\theta'}      =     \sum\limits_{[\![X\leftarrow E ]\!] \theta\models^I \phi} \Theta(\theta) [\![ Q ]\!]^I_{\theta}$ \\ $=  \sum\limits_{ \theta\models^I \phi[X/E]} \Theta(\theta) [\![ Q ]\!]^I_{\theta} = \llbracket \phi[X/E]\Rightarrow Q\rrbracket_\Theta^I $.

\item (PAS-cq) Without loss of generality, we only consider the case where $n=2$. \\

$[\![ \pt(X \xleftarrow[]{\$}  \{a_1: k_1,  ,a_2:k_2\}   ,   \phi \Rightarrow Q ) ]\!]_{\Theta}  =   [\![ a_1 (pt   (  X\leftarrow k_1,  \phi \Rightarrow Q  )  )   + a_2 (pt   (  X\leftarrow k_2,  \phi \Rightarrow Q  )  )   ]\!]_{\Theta}  $\\
 $= a_1  [\![\pt( X\leftarrow k_1,  \phi \Rightarrow Q    ) ]\!]_{\Theta}  +  a_2  [\![\pt( X\leftarrow k_2,  \phi \Rightarrow Q    ) ]\!]_{\Theta} $\\
 
 By IH, it equals $ a_1  [\![   \phi \Rightarrow Q     ]\!]_{ [\![  X\leftarrow k_1]\!]\Theta} +  a_2  [\![   \phi \Rightarrow Q     ]\!]_{ [\![  X\leftarrow k_2]\!]\Theta} $\\
 
 $=  [\![   \phi \Rightarrow Q     ]\!]_{ a_1[\![  X\leftarrow k_1]\!]\Theta} +     [\![   \phi \Rightarrow Q     ]\!]_{ a_2 [\![  X\leftarrow k_2]\!]\Theta} $\\
 
 $=  [\![   \phi \Rightarrow Q     ]\!]_{  [\![  X \xleftarrow[]{\$}  \{a_1: k_1,  a_2:k_2\} ]\!]\Theta   }$

\item (UNITARY-cq) Let $\Theta'= \llbracket U[\overline{q}]\rrbracket \Theta$. Then $\Theta'(\theta')= \sum\{ \Theta( \theta ) : \llbracket U[\overline{q}]\rrbracket \widehat{\theta} = \widehat{\theta'} \}$.\\

$[\![    \phi\Rightarrow Q       ]\!]_{  \Theta'} = \sum\limits_{\theta' \models \phi}    \Theta'(\theta')  [\![Q]\!]_{\theta'} $.\\

Note that $\theta \models [U_{\overline{q}}]\phi$ iff $\theta'\models \phi$. Then we know 
$\sum\limits_{\theta' \models \phi}    \Theta'(\theta')  [\![Q]\!]_{\theta'}   =  \sum\limits_{\theta \models [U_{\overline{q}}]\phi}    \Theta(\theta)  [\![Q]\!]_{  \llbracket U[\overline{q}]\rrbracket \theta} $.

Now we have  $[\![Q]\!]_{ \llbracket U[\overline{q}]\rrbracket \theta} = \Tr( Q  (U_{\overline{q}} \theta_q) ( U_{\overline{q}} \theta_q)^{\dagger}  )=  \Tr( Q  (U_{\overline{q}} \theta) (\theta_q^{\dagger}   U^{\dagger}_{\overline{q}}  ))  =  \Tr( U^{\dagger}_{\overline{q}} Q  (U_{\overline{q}} \theta) \theta^{\dagger}) = [\![ U^{\dagger}_{\overline{q}} Q  U_{\overline{q}} ]\!]_{\theta }  $.\\

We then know $[\![\phi\Rightarrow Q ]\!]_{ \llbracket U_{\overline{q}}\rrbracket \Theta } = \sum\limits_{\theta'\models\phi}\Theta'(\theta')\llbracket Q\rrbracket_{\theta'}=\sum\limits_{\theta\models [U_{\overline{q}}]\phi}\Theta(\theta)\llbracket Q\rrbracket_{\llbracket U[\overline{q}]\rrbracket\widehat{\theta}}=$

$\llbracket [U_{\overline{q}}]\phi\Rightarrow U_{\overline{q}}^\dagger Q U_{\overline{q}}\rrbracket_\Theta$.



\item (MEASURE-cq) 
We first want to show that, given an arbitrary $\theta\in\mathbb{PS}$, $\llbracket (\phi[X/0]\Rightarrow P_j^0 Q P_j^0)+(\phi[X/1]\Rightarrow P_j^1 Q P_j^1)\rrbracket_{\widehat{\theta}}^I=\llbracket \phi\Rightarrow Q\rrbracket _{\llbracket X\leftleftarrows q_j\rrbracket\widehat{\theta}}$.\\

We know that $[\![ X \leftleftarrows q_j  ]\!] \widehat{\theta} = p_0 \widehat{\theta_0} + p_1 \widehat{\theta_1}$, where 

$p_0 = \sqrt{ \Tr(  P^0_j \theta_q \theta_q^{\dagger} ) }$, $p_1 = \sqrt{ \Tr(  P^1_j \theta_q \theta_q^{\dagger} ) }$,  $\theta_{0,c} =  \theta_c[X\mapsto 0] $,  $\theta_{1,c} = \theta_c\theta_c[X\mapsto 1] $, $\theta_{0,q} =  \frac{P^0_j \theta_q}{\sqrt{p_0}}$, $\theta_{1,q} =  \frac{P^1_j \theta_q}{\sqrt{p_1}}$.\\

By what we have proved beforehand, $[\![ \phi[X/0] \Rightarrow  P^0_j Q P^0_j ]\!]_{ \widehat{\theta}}= [\![ \phi  \Rightarrow  P^0_j Q P^0_j ]\!]_{ [\![ X\leftarrow 0 ]\!]  \widehat{  \theta }  }  = [\![ \phi  \Rightarrow  P^0_j Q P^0_j ]\!]_{    \widehat{(\theta_{0,c}, \theta_q)}          }    $.

 Note that $\Tr( P^0_j Q P^0_j  \theta_q \theta_q^{\dagger} )=  \Tr(  Q P^0_j  \theta_q \theta_q^{\dagger} P^0_j )   = p_0 [\![Q]\!]_{ \frac{ P^0_j \theta_q}{\sqrt{p_0}}  }  = p_0 [\![Q]\!]_{ \theta_{0,q} }$.\\
 
 If $\theta \models \phi[X/0]$, then $\theta_{0,c} \models \phi $ and $[\![ \phi[X/0] \Rightarrow  P^0_j Q P^0_j ]\!]_{ \widehat{\theta}}= \Tr( P^0_j Q P^0_j  \theta_q \theta_q^{\dagger} ) = p_0 [\![Q]\!]_{ \theta_{0,q} } =   [\![  \phi \Rightarrow   Q]\!]_{p_0 \widehat{\theta_{0}}} $.

 If $\theta \not\models \phi[X/0]$, then $\theta_{0,c} \not\models \phi $ and $[\![ \phi[X/0] \Rightarrow  P^0_j Q P^0_j ]\!]_{ \widehat{\theta}}= 0 =   [\![  \phi \Rightarrow   Q]\!]_{ p_0 \widehat{\theta_{0} }}$.\\

 Similarly, we can prove that $[\![ \phi[X/1] \Rightarrow  P^1_j Q P^1_j ]\!]_{ \widehat{\theta}}  =   [\![  \phi \Rightarrow   Q]\!]_{ p_1 \widehat{\theta_{1}} }$.\\

Therefore, we have $[\![  \pt(X\leftleftarrows q_j, \phi \Rightarrow Q ) ]\!]_{\widehat{\theta}} = [\![ \phi \Rightarrow Q   ]\!]_{ [\![ X \leftleftarrows q_j  ]\!] \widehat{\theta}   } $, which implies that $[\![  \pt(X\leftleftarrows q_j, \phi \Rightarrow Q ) ]\!]_{\Theta} = [\![ \phi \Rightarrow Q   ]\!]_{ [\![ X \leftleftarrows q_j  ]\!] \Theta   } $.

\item (SEQ-cq) Trivial.
  
\item (IF-cq) $ [\![   \phi \Rightarrow Q  ]\!]_{  [\![  \texttt{if}\ B\ \texttt{then}\ C_{1}\ \texttt{else}\ C_{2}     ]\!] \Theta  }   =    [\![   \phi \Rightarrow Q  ]\!]_{  [\![  C_1     ]\!]\downarrow_{B} \Theta +   [\![  C_2     ]\!]\downarrow_{\neg B} \Theta  }     $\\
  
  $= [\![ \pt(C_1, \phi \Rightarrow Q)   ]\!]_{\downarrow_{B} \Theta  }  +    [\![ \pt(C_2, \phi \Rightarrow Q)   ]\!]_{\downarrow_{\neg B} \Theta  } $\\
  
  $= [\![  \pt(C_1, \phi \Rightarrow Q)  /B ]\!]_{\Theta}  +  [\![  \pt(C_2, \phi \Rightarrow Q)  /\neg B ]\!]_{\Theta} $\\
  
  $= [\![   (\pt(C_1, \phi \Rightarrow Q)  /B)  +(  \pt(C_2, \phi \Rightarrow Q)  /\neg B    ) ]\!]_{\Theta} $.

\item (WHILE-cq) The proof is similar to the proof of $pt(WHILE,  \mathbb{P} (\phi))$.

\end{enumerate}

\end{proof}

\begin{theorem*}\label{proof:wp-PA}
 
   $\mu \models^I \WP(C, \Phi)$ iff $ [\![C]\!]\mu \models^I \Phi$.

\end{theorem*}

\begin{proof}
We prove by structural induction:

\begin{enumerate}
\item   $\mu \models^I \WP(C, r_1  \mbox{ }  rop \mbox{ }  r_2)$ iff  $\mu \models^I \pt(C, r_1) \mbox{ }  rop \mbox{ }  \pt(C, r_2)$ iff \\
 $[\![ \pt(C, r_1)  ]\!]^I_{  \mu }  \mbox{ }  rop \mbox{ }   [\![  \pt(C, r_2) ]\!]^I_{   \mu }  = \top$ iff\\
 $[\![  r_1  ]\!]^I_{ [\![ C]\!] \mu }   \mbox{ }  rop \mbox{ }    [\![  r_2  ]\!]_{ [\![ C]\!]^I \mu }  = \top$ iff $ [\![ C]\!] \mu \models^I r_1 \mbox{ }  rop \mbox{ } r_2$.

\item $\mu \models^I  \WP(C, \neg \Phi)$ iff $ \mu \models^I \neg \WP(C, \Phi) $ iff $ \mu   \not\models^I   \WP(C, \Phi) $ iff $[\![C]\!]\mu \not\models^I \Phi$ iff $[\![C]\!]\mu  \models^I \neg \Phi$.

\item $\mu \models^I \WP(\Phi_1 \wedge \Phi_2)$ iff  $\mu \models^I \WP(C,\Phi_1) \wedge \WP(C,\Phi_2)$ iff $\mu \models^I \WP(C,\Phi_1)  $ and $\mu \models^I   \WP(C,\Phi_2)$ iff $[\![C]\!]\mu  \models^I   \Phi_1$ and $[\![C]\!]\mu  \models^I   \Phi_2$ iff $[\![C]\!]\mu  \models^I   \Phi_1 \wedge \Phi_2$.

\end{enumerate}

\end{proof}

\begin{lemma*}\label{proof:WP-syntactic}
   For an arbitrary command $C$ and an arbitrary probabilistic formula $\Phi$, it holds that $\vdash \{\WP(C,\Phi)\}C\{\Phi\}$.
\end{lemma*}

\begin{proof}

We prove it by induction on the structure of command $C$.
When $C$ is \texttt{skip}, assignment, random assignment, IF or WHILE command, we can directly derive them from the proof system PHL. The only case needs to show is the sequential command, which means that we need to prove $\vdash \{\WP(C_1;C_2,\Phi)\}C_1;C_2\{\Phi\}$.

By Definition~\ref{def.weakestpreconditionprobabilistic}, we have $\WP(C_1;C_2,\Phi)=\WP(C_1,\WP(C_2,\Phi))$. By induction hypothesis,  $\vdash\{\WP(C_2,\Phi)\}C_2\{\Phi\}$ and $\vdash\{\WP(C_1,\WP(C_2,\Phi))\}C_1\{\WP(C_2;\Phi)\}$. By inference rule $SEQ$, we can conclude that $\vdash\{\WP(C_1,\WP(C_2,\Phi))\}C_1;C_2\{\Phi\}$ which implies that $\vdash \{\WP(C_1;C_2,\Phi)\}C_1;C_2\{\Phi\}$.

\end{proof}

\section{Application}
\label{sec:application}

In this section we use our QHL to verify some quantum algorithms and protocols.

\subsection{Deutsch's algorithm} \label{sec:deutsch}

Given a boolean function $f$ with one input bit and one output bit, we want to decide whether it is a constant function. Classically, we need to evaluate the function twice (i.e., for input = 0 and input = 1), but quantumly, using Deutsch’s algorithm, we only need to evaluate the function once. Recall that the following circuit describes the operations of Deutsch’s algorithm.

\begin{center}
    \includegraphics[scale=0.15]{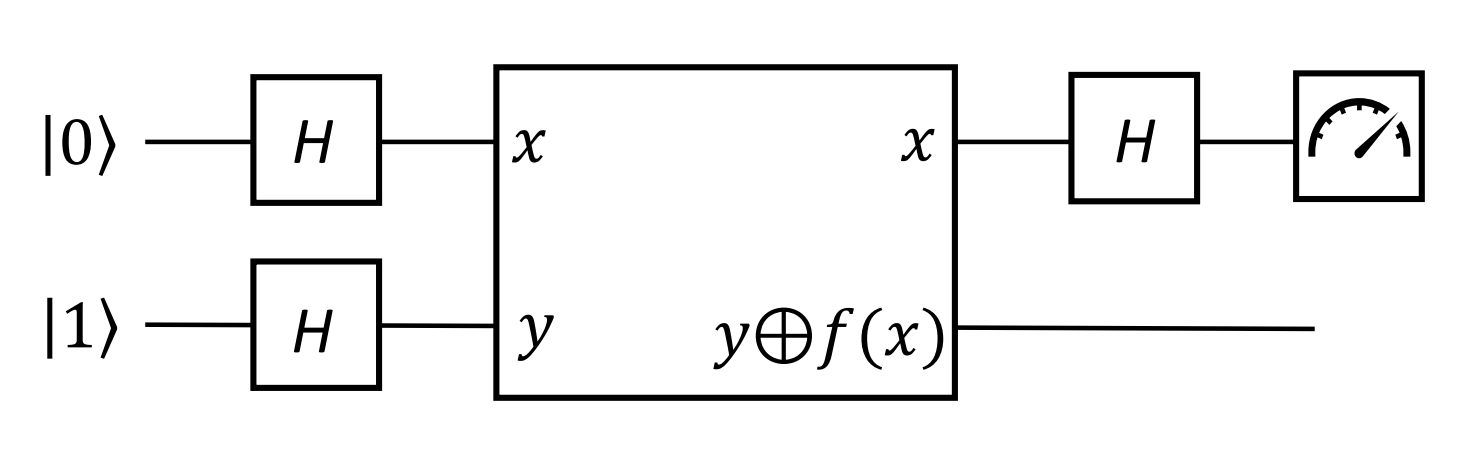}
\end{center}

\noindent 
 The middle part is $U_f$, a $2$-qubit unitary that performs $U_f\ket{x \otimes y} = \ket{ x \otimes (y \oplus f(x))}$. The algorithm conclude that $f$ is constant if and only if the measurement result is $0$. Deutsch’s algorithm is  described by the following commands in our QHL:
\[
    \text{Deutsch} := H [q_1]; H [q_2]; U_f [q_1, q_2]; H [q_1]; X \leftleftarrows q_1
\]

\noindent
We now verify Deutsch's algorithm using our QHL with deterministic assertions. Assuming $f$ is constant, we want to show   that 

\begin{center}

\htriple{$P^0_1 \land P^1_2$}{Deutsch}{$X=0$}, 

\end{center}
\noindent
i.e. if before executing Deutsch, the cq-state satisfies $q_1=\ket{0}$ and $q_2 = \ket{1}$, then after the execution of Deutsch, the cq-state satisfies $X=0$.  We will prove the above Hoare trip in a bottom-up procedure.

\tiny
\vspace{0.2cm}
\hspace*{-1.5cm}
\begin{minipage}{\textwidth}
\begin{prooftree}
\AxiomC{}
\RightLabel{(Claim1)}
\UnaryInfC{\htriple{$P^0_1 \land P^1_2$}{$H [q_1]$}{$[H_{q_1}] (P^0_1 \land P^1_2)$}}
\AxiomC{}
\RightLabel{(Claim2)}
\UnaryInfC{\htriple{$[H_{q_1}](P^0_1 \land P^1_2)$}{$H [q_2]; U_f [q_1, q_2]; H [q_1]; X\leftleftarrows q_1$}{$X=0$}}
\RightLabel{(SEQ)}
\BinaryInfC{\htriple{$P^0_1 \land P^1_2$}{$H [q_1]; H [q_2]; U_f [q_1, q_2]; H [q_1]; X\leftleftarrows q_1$}{$X=0$}}
\RightLabel{(definition)}
\UnaryInfC{\htriple{$P^0_1 \land P^1_2$}{Deutsch}{$X=0$}}
\end{prooftree}
\end{minipage}
\vspace{0.2cm}
\normalsize

\noindent
Claim 1 is due to $H$ is self-inverse, i.e. $[H_{q_1}]\,[H_{q_1}] (P^0_1 \land P^1_2) = P^0_1 \land P^1_2$, and using (UNITARY). 

\noindent
For Claim 2,

 \tiny
 \vspace{0.2cm}
\hspace*{-1.5cm}
\begin{minipage}{\textwidth}
\begin{prooftree}
\AxiomC{}
\RightLabel{(UNI)}
\UnaryInfC{\htriple{$[H_{q_1}] (P^0_1 \land P^1_2)$}{$H [q_2]$}{$[H_{q_2}]\,[H_{q_1}] (P^0_1 \land P^1_2  )$}}
\AxiomC{}
\RightLabel{(Claim3)}
\UnaryInfC{\htriple{$[H_{q_2}]\,[H_{q_1}](P^0_1 \land P^1_2)$}{$U_f [q_1, q_2]; H [q_1]; X\leftleftarrows q_1$}{$X=0$}}
\RightLabel{(SEQ)}
\BinaryInfC{\htriple{$[H_{q_1}] (P^0_1 \land P^1_2)$}{$H [q_2]; U_f [q_1, q_2]; H [q_1]; X\leftleftarrows q_1$}{$X=0$}}
\end{prooftree}
\end{minipage}
\vspace{0.2cm}
\normalsize

\noindent
For Claim 3,
 
 \tiny
 \vspace{0.2cm}
\hspace*{-2.25cm}
\begin{minipage}{\textwidth}
\begin{prooftree}
\AxiomC{}
\RightLabel{(UNI)}
\UnaryInfC{\htriple{$[H_{q_2}]\,[H_{q_1}] (P^0_1 \land P^1_2)$}{$U_f [q_1, q_2]$}{$[U^{-1}_{fq_1, q_2}]\, [H_{q_2}] \,[H_{q_1}] (P^0_1 \land P^1_2)$}}
\AxiomC{}
\RightLabel{(Claim 4)}
\UnaryInfC{\htriple{$[U^{-1}_{fq_1, q_2}]\, [H_{q_2}] \, [ H_{q_1}] (P^0_1 \land P^1_2)$}{$H [q_1]; X\leftleftarrows q_1$}{$X=0$}}
\RightLabel{(SEQ)}
\BinaryInfC{\htriple{$[H_{q_2}]\,[H_{q_1}] (P^0_1 \land P^1_2)$}{$U_f [q_1, q_2]; H [q_1]; X\leftleftarrows q_1$}{$X=0$}}
\end{prooftree}
\end{minipage}
\vspace{0.2cm}
\normalsize

\noindent
For Claim 4,

 \tiny
 \vspace{0.2cm}
\hspace*{-2.75cm}
\begin{minipage}{\textwidth}
\begin{prooftree}
\AxiomC{}
\RightLabel{(UNI)}
\UnaryInfC{\htriple{$[U^{-1}_{fq_1, q_2}]\, [H_{q_2}] \,[H_{q_1}] (P^0_1 \land P^1_2) $}{$H [q_1]$}{$[H_{q_1}] \,[U^{-1}_{fq_1, q_2}]\, [H_{q_2}] \,[H_{q_1}](P^0_1 \land P^1_2)$}}
\AxiomC{}
\RightLabel{(Claim 5)}
\UnaryInfC{\htriple{$[H_{q_1}] \,[U^{-1}_{fq_1, q_2}]\, [H_{q_2}] \,[H_{q_1}](P^0_1 \land P^1_2)$}{$X\leftleftarrows q_1$}{$X=0$}}
\RightLabel{(SEQ)}
\BinaryInfC{\htriple{$[U^{-1}_{fq_1, q_2}]\, [H_{q_2}] \,H_{q_1}](P^0_1 \land P^1_2)$}{$H [q_1]; X\leftleftarrows q_1$}{$X=0$}}
\end{prooftree}
\end{minipage}
\vspace{0.2cm}
\normalsize

\noindent
Claim 5 follows from (MEASURE) and (CONS) if we can show
 $$
 \models [H_{q_1}] \,[U^{-1}_{fq_1, q_2}]\, [H_{q_2}] \,[H_{q_1}] (P^0_1 \land P^1_2) \rightarrow (  [ \mathfrak{P}^0_1 ] (X=0)[X/0] \vee P^1_1 ) \wedge   ( [ \mathfrak{P}^1_1 ] (X=0)[X/1] \vee P^0_1 ).  $$

\noindent
Therefore, we need to show that
\[
 \sem{[H_{q_1}] \,[U^{-1}_{fq_1, q_2}]\, [H_{q_2}] \,[H_{q_1}] ( P^0_1 \land P^1_2)} \subseteq \sem{(  [ \mathfrak{P}^0_1 ] (X=0)[X/0] \vee P^1_1 ) \wedge   ( [ \mathfrak{P}^1_1 ] (X=0)[X/1] \vee P^0_1 )}
\]
The right-hand side reduces to
\[
\sem{(  [ \mathfrak{P}^0_1 ] (X=0)[X/0] \vee P^1_1 ) \wedge   P^0_1} = \sem{  [ \mathfrak{P}^0_1 ] (X=0)[X/0] \land P^0_1} = \sem{P^0_1}
\]
To reduce the left-hand side, we do a case analysis. 
If $f$ is a constant function which maps all input to $0$, then $U_f$ is the identity operator. The left-hand side reduces to 
\[
\sem{[H_{q_1}] \, [H_{q_2}] \,[H_{q_1}] (P^0_1 \land P^1_2)} = \sem{[H_{q_2}] ( P^0_1 \land P^1_2)}= \sem{P^0_1 \land [H_{q_2}]P^1_2} \subseteq \sem{P^0_1}.
\]
If $f$  is a constant function which maps all input to $1$, then $U_f = I \otimes  X$, where $I$ is identity and $ X$ is the X-gate. The left-hand side reduces to 
\[
\sem{[H_{q_1}] \, [H_{q_2}]\, [I \otimes  X]\,[H_{q_1}] (P^0_1 \land P^1_2)} = \sem{[H_{q_2}]\,[X_{q_2}] (P^0_1 \land P^1_2)}= \sem{P^0_1 \land [H_{q_2}]\,[X_{q_2}]P^1_2} \subseteq \sem{P^0_1}.
\]

This concludes the proof of
 $$
 \models [H_{q_1}] \,[U^{-1}_{fq_1, q_2}]\, [H_{q_2}] \,[H_{q_1}] (P^0_1 \land P^1_2) \rightarrow (  [ \mathfrak{P}^0_1 ] (X=0)[X/0] \vee P^1_1 ) \wedge   ( [ \mathfrak{P}^1_1 ] (X=0)[X/1] \vee P^0_1 ).  $$
The case when $f$ is balanced is similar.

\subsection{Quantum teleportation}

Quantum teleportation is a technique for transferring a qubit from a sender at one location to a receiver at another location. The following is the  circuit representation of quantum teleportation. The circuit consumes a qubit $\ket{\psi}$ at a location and produce $\ket{\psi}$ at another location. The sender and receiver share a Bell state, each having one qubit. The first two gates acting on $\ket{0}\otimes \ket{0}$ is to setup the Bell state $\frac{1}{\sqrt 2}(\ket{00}+\ket{11})$. Then the sender performs a controlled-not gate on $\ket{\psi}$ and a qubit from Bell state, a Hardmard gate on $\ket{\psi}$, and then measurements both qubits and sends the result of measurements to the receiver. Lastly, the receiver performs correction according to the measurement result.

\begin{center}
    \includegraphics[scale=0.25]{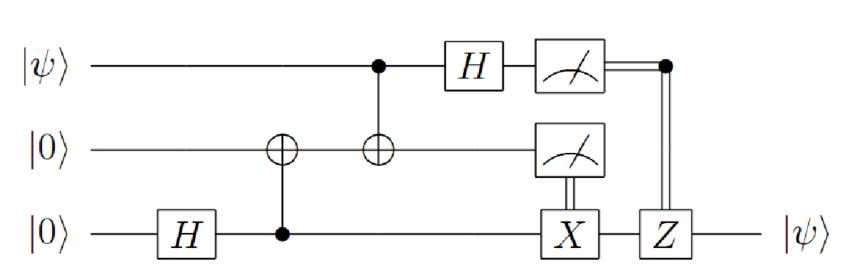}
\end{center}

The operations can be described by the following command:
\begin{equation*} \label{eq2}
\begin{split}
    & \text{Tele} := H [q_3]; CX [q_3,q_2]; CX [q_1, q_2]; H [q_1]; M_1 \leftleftarrows q_1; M_2 \leftleftarrows q_2; \\
    & \text{if } M_2 =1  \text{ then } X[q_3] \text{ else skip}; \text{if } M_1 =1 \text{ then } Z[q_3] \text{ else skip}
\end{split}
\end{equation*}
where $X,Z,CX$ are gates. We use $\text{Tele}_1;\text{Tele}_2$ to denote the first and second line of the  command.

 Assuming $\ket{\psi} = \ket{0}$, we now use our proof system to prove  
 \begin{center}
  \htriple{$P^0_1 \land P^0_2 \land P^0_3$}{Tele}{$P^0_3$}
 
 \end{center}
in a bottom-up procedure.

\newpage 

\begin{prooftree}
\AxiomC{}
\RightLabel{(Claim1)}
\UnaryInfC{\htriple{$P^0_1 \land P^0_2 \land P^0_3$}{$\text{Tele}_1$}{$V$}}
\AxiomC{}
\RightLabel{(Claim2)}
\UnaryInfC{\htriple{$V$}{$\text{Tele}_2$}{$P^0_3$}}
\RightLabel{(SEQ)}
\BinaryInfC{\htriple{$P^0_1 \land P^0_2 \land P^0_3$}{$\text{Tele}_1;\text{Tele}_2$}{$P^0_3$}}
\RightLabel{(definition)}
\UnaryInfC{\htriple{$P^0_1 \land P^0_2 \land P^0_3$}{Tele}{$P^0_3$}}
\end{prooftree}
\normalsize

For Claim 2, 

\tiny
 \vspace{0.2cm}
\hspace*{-2.5cm}
\begin{minipage}{\textwidth}
\begin{prooftree}
\AxiomC{}
\RightLabel{(2a)}
\UnaryInfC{\htriple{$V \land (M_2 =1)$}{$X[q_3]$}{W}}
\AxiomC{}
\RightLabel{(2b)}
\UnaryInfC{\htriple{$V \land \lnot (M_2 =1)$}{$skip$}{W}}
\RightLabel{(IF)}
\BinaryInfC{\htriple{$V$}{$\text{if } M_2 =1 \text{ then } X[q_3] \text{ else skip}$}{$W$}}
\AxiomC{}
\RightLabel{(2c)}
\UnaryInfC{\htriple{$W \land (M_1 =1)$}{$Z[q_3]$}{$P^0_3$}}
\AxiomC{}
\RightLabel{(2d)}
\UnaryInfC{\htriple{$W \land \lnot (M_1 =1 )$}{$skip$}{$P^0_3$}}
\RightLabel{(IF)}
\BinaryInfC{\htriple{$W$}{$\text{if } M_1 =1 \text{ then } Z[q_3] \text{ else skip}$}{$P^0_3$}}
\RightLabel{(SEQ)}
\BinaryInfC{\htriple{$V$}{$\text{if } M_2 =1 \text{ then } X[q_3] \text{ else skip}; \text{if } M_1 =1 \text{ then } Z[q_3] \text{ else skip}$}{$P^0_3$}}
\end{prooftree}
\end{minipage}
\vspace{0.2cm}
\normalsize

\noindent
where $W=( (M_1 =1) \land [Z_{q_3}]P_3^0)\lor (\lnot (M_1 =1 ) \land P_3^0)$, and $  V = ((M_2 =1) \land [X_{q_3}]W)\lor (\lnot ( M_2 =1) \land W)$. Here 2b and 2d is by (SKIP) and (CON), 2a and 2c is by  (CON) and (UNITARY).

For Claim 1, 

\tiny
 \vspace{0.2cm}
\hspace*{-2cm}
\begin{minipage}{\textwidth}
\begin{prooftree}
\AxiomC{}
\RightLabel{(UNI,SEQ)}
\UnaryInfC{\htriple{$P^0_1 \land P^0_2 \land P^0_3$}{$H [q_3]; CX [q_3,q_2]; CX [q_1, q_2]; H [q_1]$}{$[H_{q_1}][CX_{q_3,q_2}][CX_{q_1,q_2}][H_{q_3}] ( P^0_1 \land P^0_2 \land P^0_3 )$}}
\AxiomC{}
\RightLabel{(Claim3)}
\UnaryInfC{H3}
\RightLabel{(SEQ)}
\BinaryInfC{\htriple{$P^0_1 \land P^0_2 \land P^0_3$}{$H [q_3]; CX [q_3,q_2]; CX [q_1, q_2]; H [q_1]; M_1 \leftleftarrows q_1; M_2 \leftleftarrows q_2$}{$V$}}
\end{prooftree}
\end{minipage}
\vspace{0.2cm}
\normalsize

where $H3 =$  \htriple{$[H_{q_1}][CX_{q_3,q_2}][CX_{q_1,q_2}][H_{q_3}]  ( P^0_1 \land P^0_2 \land P^0_3 )$}{$M_1 \leftleftarrows q_1; M_2 \leftleftarrows q_2$}{$V$}. For Claim 3,

\tiny
\begin{prooftree}
\AxiomC{}
\RightLabel{(Claim4)}
\UnaryInfC{H4}
\AxiomC{}
\RightLabel{(MEAS)}
\UnaryInfC{\htriple{$([ \mathfrak{P}^0_2 ] V[M_2/0] \vee P^1_2 ) \wedge   ( [ \mathfrak{P}^1_2 ] V[M_2/1] \vee P^0_2 )$}{ $M_2 \leftleftarrows q_2$}{$V$} }
\RightLabel{(SEQ)}
\BinaryInfC{\htriple{$[H_{q_1}][CX_{q_3,q_2}][CX_{q_1,q_2}][H_{q_3}] ( P^0_1 \land P^0_2 \land P^0_3)$}{$M_1 \leftleftarrows q_1; M_2 \leftleftarrows q_2$}{$V$}}
\end{prooftree}
\normalsize
where $H4 =$  \htriple{$H[q_1]CX[q_1,q_2]CX[q_3,q_2]H[q_3] ( P^0_1 \land P^0_2 \land P_3^0 )$}{$M_1 \leftleftarrows q_1$}{$([ \mathfrak{P}^0_2 ] V[M_2/0] \vee P^1_2 ) \wedge   ( [ \mathfrak{P}^1_2 ] V[M_2/1] \vee P^0_2 )$}.
For Claim4, in order to use (MEAS), we need to prove

$$ \models  [H_{q_1}][CX_{q_1,q_2}][CX_{q_3,q_2}][H_{q_3}] (P^0_1 \land P^0_2 \land P^0_3) \rightarrow ([ \mathfrak{P}^0_1 ] \phi [M_1/0] \vee P^1_1 ) \wedge   ( [ \mathfrak{P}^1_1 ] \phi[M_1/1] \vee P^0_1 ) $$
 
where 
\begin{equation*} \label{eq3}
\begin{split}
    \phi &= ([ \mathfrak{P}^0_2 ] V[M_2/0] \vee P^1_2 ) \wedge   ( [ \mathfrak{P}^1_2 ] V[M_2/1] \vee P^0_2 )\\
    &= ([ \mathfrak{P}^0_2 ] (((M_2 =1) \land [X_{q_3}]W)\lor (\lnot (M_2=1) \land W))[M_2/0] \vee P^1_2 ) \wedge   ( [ \mathfrak{P}^1_2 ] V[M_2/1] \vee P^0_2 )\\
    &= ([ \mathfrak{P}^0_2 ] W \vee P^1_2 ) \wedge   ( [ \mathfrak{P}^1_2 ] [X_{q_3}]W \vee P^0_2 )
\end{split}
\end{equation*}

Note that here we abuse the notation $=$ to represent logical equivalence. Claim 4 follows from the calculation below,

$ ([ \mathfrak{P}^0_1 ] \phi [M_1/0] \vee P^1_1 ) \wedge   ( [ \mathfrak{P}^1_1 ] \phi[M_1/1] \vee P^0_1 )  $

$= ([ \mathfrak{P}^0_1 ] (([ \mathfrak{P}^0_2 ] W \vee P^1_2 ) \wedge   ( [ \mathfrak{P}^1_2 ] [X_{q_3}]W \vee P^0_2 )) [M_1/0] \vee P^1_1 ) \wedge   ( [ \mathfrak{P}^1_1 ] \phi[M_1/1] \vee P^0_1 ) $

$  =([ \mathfrak{P}^0_1 ] (([ \mathfrak{P}^0_2 ] (( ( M_1 =1) \land [Z_{q_3}]P_3^0)\lor (\lnot ( M_1 =1) \land P_3^0)) \vee P^1_2 )   \wedge   ( [ \mathfrak{P}^1_2 ] [X_{q_3}]W \vee P^0_2 )) [M_1/0] \vee P^1_1 ) \wedge   ( [ \mathfrak{P}^1_1 ] \phi[M_1/1] \vee P^0_1 ) $

$=  ([ \mathfrak{P}^0_1 ] (([ \mathfrak{P}^0_2 ] (( (0 =1) \land [Z_{q_3}]P_3^0)\lor (\lnot ( 0 =1) \land P_3^0)) \vee P^1_2 )   \wedge   ( [ \mathfrak{P}^1_2 ] [X_{q_3}]W \vee P^0_2 )) [M_1/0] \vee P^1_1 ) \wedge   ( [ \mathfrak{P}^1_1 ] \phi[M_1/1] \vee P^0_1 ) $

 $=  ([ \mathfrak{P}^0_1 ] (([ \mathfrak{P}^0_2 ] (       P_3^0        \vee P^1_2 )   \wedge   ( [ \mathfrak{P}^1_2 ] [X_{q_3}]W \vee P^0_2 )) [M_1/0] \vee P^1_1 ) \wedge   ( [ \mathfrak{P}^1_1 ] \phi[M_1/1] \vee P^0_1 ) $

  $     =  ([ \mathfrak{P}^0_1 ] (([ \mathfrak{P}^0_2 ] (       P_3^0        \vee P^1_2 )   \wedge   ( [ \mathfrak{P}^1_2 ] [X_{q_3}] (   ( (M_1 =1) \land [Z_{q_3}]P_3^0)\lor (\lnot (M_1 =1 ) \land P_3^0))                                                               
  \vee P^0_2 )) [M_1/0] \vee P^1_1 ) \wedge   ( [ \mathfrak{P}^1_1 ] \phi[M_1/1] \vee P^0_1 )        $

  $     =  ([ \mathfrak{P}^0_1 ] (([ \mathfrak{P}^0_2 ] (       P_3^0        \vee P^1_2 )   \wedge   ( [ \mathfrak{P}^1_2 ] [X_{q_3}] (   ( (0 =1) \land [Z_{q_3}]P_3^0)\lor (\lnot (0 =1 ) \land P_3^0))                                                               
  \vee P^0_2 ))  \vee P^1_1 ) \wedge   ( [ \mathfrak{P}^1_1 ] \phi[M_1/1] \vee P^0_1 )        $

    $     =  ([ \mathfrak{P}^0_1 ] (([ \mathfrak{P}^0_2 ] (       P_3^0        \vee P^1_2 )   \wedge   ( [ \mathfrak{P}^1_2 ] [X_{q_3}]        ( P_3^0 \vee P^0_2  )))  \vee P^1_1 ) \wedge   ( [ \mathfrak{P}^1_1 ] \phi[M_1/1] \vee P^0_1 )        $

$ =([ \mathfrak{P}^0_1 ] (([ \mathfrak{P}^0_2 ] P_3^0 \vee P^1_2 ) \wedge   ( [ \mathfrak{P}^1_2 ] [X_{q_3}] (P_3^0 \vee P^0_2) )) \vee P^1_1 ) \wedge   ( [ \mathfrak{P}^1_1 ] (([ \mathfrak{P}^0_2 ] ( (M_1 =1) \land [Z_{q_3}]P_3^0)\lor ( \lnot( M_1=1) \land P_3^0) \vee P^1_2 ) \wedge   ( [ \mathfrak{P}^1_2 ] [X_{q_3}]W \vee P^0_2 )) [M_1/1] \vee P^0_1 )  $

$  =([ \mathfrak{P}^0_1 ] (([ \mathfrak{P}^0_2 ]  (P_3^0 \vee P^1_2 )) \wedge   ( [ \mathfrak{P}^1_2 ] [X_{q_3}]  (P_3^0 \vee P^0_2) )) \vee P^1_1 ) \wedge   ( [ \mathfrak{P}^1_1 ] (([ \mathfrak{P}^0_2 ] [Z_{q_3}]  ( P_3^0 \vee P^1_2 )  ) \wedge   ( [ \mathfrak{P}^1_2 ] [X_{q_3}][Z_{q_3}]  (  P_3^0 \vee P^0_2) )) \vee P^0_1 )  $

$  =([ \mathfrak{P}^0_1 ] (([ \mathfrak{P}^0_2 ] ( P_3^0 \vee P^1_2) ) \wedge   ( [ \mathfrak{P}^1_2 ] ( P_3^1 \vee P^0_2) )) \vee P^1_1 ) \wedge  $\\ $ ( [ \mathfrak{P}^1_1 ] (([ \mathfrak{P}^0_2 ] ( P_3^0 \vee P^1_2) ) \wedge   ( [ \mathfrak{P}^1_2 ] (P_3^1 \vee P^0_2) )) \vee P^0_1 )  =: \chi $

which is implied by
$    [H_{q_1}][CX_{q_3,q_2}][CX_{q_1,q_2}][H_{q_3}] (P^0_1 \land P^0_2 \land P^0_3 ) $.
The implication holds because for all pure cq-state $\theta$,

if $\theta \in \sem{[H_{q_1}][CX_{q_3,q_2}][CX_{q_1,q_2}][H_{q_3}] (P^0_1 \land P^0_2 \land P^0_3 )}  $,

then $ H_{q_3}  CX_{q_1,q_2}    CX_{q_3,q_2}  H_{q_1}  \theta  \in \sem{ (P^0_1 \land P^0_2 \land P^0_3 )} =\{   \theta' | \theta'_q  = |000\rangle \} $,

then $      \theta  \in  \{   \theta' | \theta'_q  =    H_{q_1}  CX_{q_3,q_2}  CX_{q_1,q_2} H_{q_3} |000\rangle \} =   \{   \theta' | \theta'_q  =   \frac{1}{{2}}(\ket{000} + \ket{100} + \ket{011}+ \ket{111}) \} $.

That is, $ \sem{[H_{q_1}][CX_{q_3,q_2}][CX_{q_1,q_2}][H_{q_3}] (P^0_1 \land P^0_2 \land P^0_3 )} \subseteq    \sem{\chi } $

The case for $\ket{\phi} = \ket{1}$ is similar. By linearity, we have shown that teleportation works using QHL.

\end{document}